\documentclass[11pt]{article}

\usepackage{authordate1-4}
\usepackage{FRdina4}
\usepackage{amssymb}  
\usepackage{MyMacros}
\usepackage{color}
\usepackage{url}

\begin{document}
\pagestyle{headings}

\title{\textbf{ Probabilistic Argumentation and Information Algebras of Probability Potentials on Families of Compatible Frames}}

\author{Juerg Kohlas \\
\small Department of Informatics DIUF \\ 
\small University of Fribourg \\ 
\small CH -- 1700 Fribourg (Switzerland) \\ 
\small E-mail: \texttt{juerg.kohlas@unifr.ch} \\
\small \texttt{http://diuf.unifr.ch/drupal/tns/juerg\_kohlas
}
}
\date{\today}

\maketitle


\begin{abstract}
Probabilistic argumentation is an alternative to causal modeling with Bayesian networks. Probabilistic argumentation structures (PAS) are defined on families of compatible frames (f.c.f). This is a generalization of the usual multivariate models based on families of variables. The crucial relation of conditional independence between frames of a f.c.f is introduced and shown to form a quasi-separoid, a weakening of the well-known structure of a separoid. It is shown that PAS generate probability potentials on the frames of the f.c.f. The operations of aggregating different PAS and of transport of a PAS from one frame to another induce an algebraic structure on the family of potentials on the f.c.f, an algebraic structure which is similar to valuation algebras related to Bayesian networks, but more general. As a consequence the well-known local computation architectures of Bayesian networks for inference apply also for the potentials on f.c.f. Conditioning and conditionals can be defined for potentials and it is shown that these concepts satisfy similar properties as conditional probability distributions. Finally a max/prod algebra between potentials is defined and applied to find most probable configurations for a factorization of potentials. 
\end{abstract}

\tableofcontents


\section{Introduction}

Probabilistic reasoning is usually associated with causal modellng and Bayesian networks. In essence, a multidimensional probability distribution is factorized into a product of prior probability distributions and conditional distributions. Information coming from observing certain events is combined with this distribution using Bayes theorem. Local computation architectures are used to compute efficiently marginal distributions of interest \cite{lauritzenspiegelhalter88}. This is based on an underlying algebraic structure called valuation algebras, which cover also many other uncertainty formalisms \cite{shenoyshafer90,kohlasshenoy00,kohlas03}. There is a vast literature on the subject of probabilistic or Bayesian networks. However, there is an alternative approach to probabilistic modeling which is much less known. Probabilistic argumentation is based on the idea that uncertain information depends on unknown assumptions, which however are more or less likely or probable. This concept is developed in this paper.

The mathematical structure underlying probabilistic argumentation, as understood in this paper, is introduced in Section \ref{sec:PAS}. It is explained how this structure can be used to evaluate hypotheses about unknown elements; and it is also shown how different probabilistic argumentation structures (PAS) can be combined or aggregated. This approach is extended in Section \ref{sec:FcF}: Probabilistic information may concern \textit{different}, but usually \textit{connected} or \textit{related} questions. This is captured by the concept of \textit{families of compatible frames} (f.c.f), a concept borrowed from \cite{shafer76} and adapted to the needs of the present theory. Families of compatible frames cover especially the case of lattices of partitions of an universe or the popular multivariate model, used exclusively for Bayesian networks, as special cases. The crucial notion is the one of conditional independence between frames of a f.c.f. This concept relative to f.c.f has already been discussed in \cite{kohlasmonney95}. Here it is shown that it induces a structure called \textit{quasi-separoid} (or q-spearoid), a weakening of the well-known structure of a separoid \cite{Dawid2001}, a mathematical framework for conditional independence. Q-separoids turn out to be basic for an algebraic structure associated with PAS.

The operation of combination of PAS is extended to the aggregation of PAS related to different frames of an f.c.f; and in addition a new operation of extraction of information from a PAS relative to a coarser frame is introduced (see Section \ref{sec:ReasPAS}). This gives rise to an algebraic structure of probability potentials associated with probabilistic argumentation structures (Sections \ref{subsec:AlgOfProbPot}, \ref{subsec:CommFcF} and \ref{subsec:AbsContPAS}). In fact, it is an algebraic structure embedded in some precise sense into an instance of an information algebra as introduced in \cite{kohlas17}. The surrounding information algebra, in which the algebra of probability potentials is embedded in essentially an algebra of belief functions or set potentials in the sense of Dempster-Shafer theory \cite{shafer76,kohlas03} In the special case that the f.c.f corresponds to a multivariate model, this algebra is identical to the well-known valuation algebra underlying probability propagation in Bayesian networks \cite{shenoyshafer90}. This important algebra has therefore a new, and much more direct and natural  interpretation than the usual one related to Bayesian networks. Not every combination of probability potentials represent a joint probability distribution as a product of conditional probability distributions as in causal modeling. That is, probability potentials represent therefore more general information. This algebraic structure is sufficient to extend the known local computation architectures from the multivariate case to the more general case of potentials on f.c.f. This is discussed in Section \ref{sec:LocComp}. 
 
Conditional probability distributions are usually defined with respect to sets of variables, that is in a multivariate model of variables. However, probability potentials may be defined relative to more general universes, like f.c.f, as shown in this paper. In Section \ref{sec:Cond} conditioning is studied in this more general framework. It is shown that conditionals may be defined in this general context too, and that they have similar properties as usual conditional distributions. This is an instance of more general conditionals defined relative to some abstract valuation algebras, namely regular valuation algebras, see \cite{kohlas03}.

Finally, in Section \ref{sec:MostProb},  the problem of finding the most probable configuration of a factorized potential is examined. It turns out that the corresponding max/prod algebra of probability potentials is exactly an information algebra on the f.c.f in the sense of \cite{kohlas17}. Thus, similar to \cite{shenoy91b,shenoy96}, architectures of dynamic programming combined with local computation are possible in our general framework based on families of compatible frames.

To conclude this introduction, we remark that probabilistic argumentation may be applied to more general concepts than probability potentials. In particular, we mention hints and belief functions, as examples where an interpretation as probability interpretation in our sense is possible. In fact, probabilistic argumentation structures as defined in this paper are special (namely precise) hints. Further, probabilistic argumentation may be used in the framework of logic \cite{kohlas03b}, especially propositional logic \cite{haennikohlas00,kohlasmoral96} and also for statistical inference, see \cite{monney00,km04,kohlasmonney07}. A more theoretical general analysis of probabilistic argumentation can be found in \cite{kohlas07rv}.


\section{Probabilistic Argumentation Structures} \label{sec:PAS}

To start, we define what we understand by a \textit{probabilistic argumentation structure}. Let $\Theta$ be a finite set, whose elements are thought of representing possible answers to some given question. We imagine that under some assumptions, the answer to this question is given or known. So, let $\Omega$ be a finite set, whose elements represent different assumptions. For any assumption $\omega \in \Omega$, the answer to the question  is given by $X(\omega) \in \Theta$. Now, the true assumption, the assumption which is actually valid, may be unknown. But a probability distribution $p(\omega)$ over $\Omega$ will describe the likelihood of the different possible assumptions. These elements together form what we call a probabilistic argumentation structure.

\begin{definition} \textit{Probabilistic Argumentation Structure (PAS):}
If $\Theta$ is a finite set, $(\Omega,p)$ a discrete probability space, that is, $\Omega$ a finite set and $p : \Omega \rightarrow [0,1]$ such that 
\begin{eqnarray*}
0 \leq p(\omega) \leq 1, \quad \sum_{\omega \in \Omega} p(\omega) = 1,
\end{eqnarray*}
and $X : \Omega \rightarrow \Theta$ a mapping from $\Omega$ into $\Theta$, then the quatrupel $(\Omega,p,X,\Theta)$ is called a probabilistic argumentation structure for $\Theta$. The elements of $\Omega$ are called assumptions.
\end{definition}

A PAS is a piece of information, which allows to judge and evaluate hypotheses about the unknown answer in $\Theta$. For instance, for an element $\theta \in \Theta$, we may ask what are the arguments that $\theta$ is the unknown answer to the question? Any assumption $\omega$ such that $X(\omega) = \theta$ is an argument in favour of $\theta$: If such an assumption happens to be true, then $\theta$ is the answer looked for. The set $s(\theta) = \{\omega \in \Omega:X(\omega) = \theta\}$ contains all arguments in favour of $\theta$, it is called the support set of $\theta$. The probability of this set
\begin{eqnarray*}
sp(\theta) = P(s(\theta)) = \sum_{\omega:X(\omega) = \theta} p(\omega)
\end{eqnarray*}
measures the degree of support for $\theta$ in the PAS $(\Omega,p,X,\Theta)$. We have obviously
\begin{eqnarray*}
0 \leq sp(\theta) \leq 1, \quad \sum_{\theta \in \Theta} sp(\theta) = 1.
\end{eqnarray*}
So, $sp(\theta)$ defines a probability distribution over $\Theta$, induced by the PAS $(\Omega,p,X,\Theta)$. Also, if $T$ is any subset of $\Theta$, then $s(T) = \{\omega \in \Omega:X(\omega) \in T\}$ contains all arguments in favour of an answer in $T$. And 
\begin{eqnarray*}
sp(T) = P(s(T)) = \sum_{\theta \in T} sp(\theta)
\end{eqnarray*}
is the degree of support of the set $T$.

This is the essence of probabilistic argumentation in the simple framework of a PAS. Now, there may be two or more PAS given relative to $\Theta$. How can these be combined? Let's consider two PAS $(\Omega_1,p_1,X_1,\Theta)$ and $(\Omega_2,p_2,X_2,\Theta)$. If $\omega_1$ is a possible assumption in the first PAS and $\omega_2$ a possible assumption in the second PAS, then the pair $(\omega_1,\omega_2)$ represents the joint assumption for the two PAS together. However, when $X_1(\Omega_1) \not= X_2(\omega_2)$, then the two individual assumptions contradict each other and can not jointly be true. So let
\begin{eqnarray*}
\Omega = \{(\omega_1,\omega_2) \inÊ\Omega_1 \times \Omega_2:X_1(\omega_1) = X_2(\omega_2)\}
\end{eqnarray*}
denote the jointly possible or consistent assumptions of the two PAS. Can we assign probabilities to the elements of $\Omega$ using the probability distribution $p_1$ and $p_2$ of the two original PAS? We may consider any probability distribution $p$ in $\Omega_1 \times \Omega_2$ such that the marginal distributions relative to $\Omega_1$ and $\Omega_2$ are exactly $p_1$ and $p_2$,
\begin{eqnarray*}
\sum_{\omega_2 \inÊ\Omega_2} p(\omega_1,\omega_2) = p_1(\omega_1), \quad
\sum_{\omega_1 \inÊ\Omega_1} p(\omega_1,\omega_2) = p_2(\omega_2).
\end{eqnarray*}
There are of course many probability distributions $p$ satisfying this consistency requirement. The simplest case arises, if we assume that the assumptions in the two PAS are a priori \textit{stochastically independent} such that
\begin{eqnarray*}
p(\omega_1,\omega_2) = p_1(\omega_1) \cdot p_2(\omega_2).
\end{eqnarray*}
It is a question of modeling to determine $p$. Subsequently, we shall adopt the independence assumption. For different alternative modelling approaches we refer to \cite{kohlas03,haennikohlas00,poulykohlas11}. In Section \ref{subsec:AbsContPAS} a useful model will help to justify, why this independence assumption is often justified. 

Now, a posteriori, if we accept that only elements of $\Omega$ are jointly possible assumptions, we must condition $p$ on the event $\Omega$ such that we obtain
\begin{eqnarray*}
p'(\omega_1,\omega_2) = k^{-1} \cdot p(\omega_1,\omega_2), \textrm{ where}\ k = \sum_{(\omega_1,\omega_2) \in \Omega} p(\omega_1,\omega_2).
\end{eqnarray*}
Finally, define $X(\omega_1,\omega_2) = X_1(\omega_1) = X_2(\omega_2)$ on $\Omega$. Then, the combined PAS, obtained from the two original PAS is defined as $(\Omega,p,X,\Theta)$. So, this is a possible approach to aggregate individual PAS into a new combined PAS, leading to a kind of algebra of PAS. This point of view will subsequently be reconsidered and worked out in a more general framework.


\section{Family of Compatible Frames} \label{sec:FcF}

\subsection{Compatible Frames} \label{subsec:CompFrames}

We extend now our discussion by considering not only a fixed frame $\Theta$, but a whole family of interrelated frames and PAS relative to the frames of this family. The proper context for this is the concept of \textit{families of compatible frames} (f.c.f) as introduced in \cite{shafer76}. The basic idea is that a frame $\Theta$ (a finite set) may be refined by splitting the elements of $\Theta$ into several (finer) elements, which form a new frame, say $\Lambda$. This concept is mathematically seized by a map $\tau : \Theta \rightarrow 2^\Lambda$ (where $2^\Lambda$ denotes the power set of $\Lambda$) which assigns to each element of $\Theta$ a subset $\tau(\theta) \subseteq \Lambda$ such that
\begin{enumerate}
\item $\tau(\theta) \not= \emptyset$ for all $\theta \in \Theta$,
\item $\tau(\theta') \cap \tau(\theta'') =Ê\emptyset$ for $\theta' \not= \theta''$,
\item $\cup_{\theta \in \Theta} \tau(\theta) = \Lambda$.
\end{enumerate}
Such a map $\tau$ is called a \textit{refining} of $\Theta$, the set $\Lambda$ a \textit{refinement} of $\Theta$ and the latter a \textit{coarsening} of the former. Note that the sets $\tau(\theta)$ form a \textit{partition} of $\Lambda$. A refining can be extended to a map of sets,
\begin{eqnarray*}
\tau(S) = \cup_{\theta \in S} \tau(\theta)
\end{eqnarray*}
for any subset $S$ of $\Theta$.

To a refining $\tau : \Theta \rightarrow 2^\Lambda$ we assign a map $v : 2^\Lambda \rightarrow 2^\Theta$  defined by
\begin{eqnarray} \label{eq:DefOfSat}
v(S) = \{\theta \in \Theta:\tau(\theta) \cap S \not= \emptyset\},
\end{eqnarray}
defined for any subset $S$ of $\Lambda$. This is called a \textit{saturation map}. We may interpret frames $\Theta$ as sets of possible answers to a question. In a refinement $\Lambda$ of $\Theta$ each possible answer $\theta$ to the first question represented by $\Theta$ is split into a set of finer possible answers to a finer question represented by $\Lambda$. Conversely, sets $\tau(\theta)$ of possible answers in $\Lambda$ are collected into a coarser answer $\theta$, a possible answer of the coarser question represented by $\Theta$. Then the elements of the set $\tau(\theta)$ are all possible answers in $\Lambda$ compatible with the $\theta$ in $\Theta$ and $v(\{\lambda\})$ represents all possible answers in $\Theta$ compatible with the element $\lambda$ in $\Lambda$. Or, more generally, $v(S)$ contains all possible answers in $\Theta$ compatible with some element in subset $S$ of $\Lambda$. This point of view will be worked out below in a more general way.

The elements of a family of questions to be considered must be related to each other in some way. This is captured by the concept of a \textit{family of compatible frames}.

\begin{definition} \label{def:fcf} \textit{Family of Compatible Frames:} 
A pair $(\mathcal{F},\mathcal{R})$ of frames and refinings $\mathcal{R}$ between frames of $\mathcal{F}$ is called a family of compatible frames (f.c.f) provided the following conditions are satisfied:
\begin{enumerate}
\item \textit{Composition of Refinings:} If, for $\Theta_1,\Theta_2,\Theta_3 \in \mathcal{F}$, $\tau_{1} :\Theta_{1} \rightarrow 2^{\Theta_2}$ and $\tau_{2} : \Theta_{2} \rightarrow 2^{\Theta_3}$ belong to $\mathcal{R}$, then $\tau_{2} \circ \tau_{1} \in \mathcal{R}$.
\item \textit{Identity:} If $\Theta \in \mathcal{F}$, then the identity map $id :  \Theta \rightarrow  2^\Theta$, defined by $id(\theta) = \{\theta\}$, belongs to $\mathcal{R}$.
\item \textit{Identity of Refinings:} If $\tau_{1} :  \Theta \rightarrow  2^\Lambda$ and $\tau_{2} :  \Theta \rightarrow  2^\Lambda$ are elements of $\mathcal{R}$, then $\tau_{1} = \tau_{2}$.
\item \textit{Identity of Coarsenings:} If $\tau_{1} : \Theta_{1} \rightarrow 2^\Lambda$ and $\tau_{2} : \Theta_{2} \rightarrow 2^\Lambda$ belong to $\mathcal{R}$ and if for each $\theta_{2} \in \Theta_{2}$ there exists a $\theta_{1} \in \Theta_{1}$ and for each $\theta_{1} \in \Theta_{1}$ there exists a $\theta_{2} \in \Theta_{2}$ such that $\tau_{1}(\theta_{1}) = \tau_{2}(\theta_{2})$, then $\Theta_{1} = \Theta_{2}$.
\item \textit{Existence of Minimal Common Refinement:} For any finite family $\Theta_{1},\ldots,\Theta_{n}$ of frames in $\mathcal{F}$, there exists a common refinement $\Lambda \in \mathcal{F}$ such that if $\Lambda' \in \mathcal{F}$ is another common refinement of $\Theta_{1},\ldots,\Theta_{n}$, then $\Lambda'$ is also a refinement of $\Lambda$, and, if $\tau_i$ are the refinings of $\Theta_i$ to $\Lambda$, then for every $\lambda \in \Lambda$, there exist elements $\theta_i \in \Theta_i$ such that
\begin{eqnarray} \label{eq:SuffQSep}
\tau_1(\theta_1) \cap \ldots \cap \tau_n(\theta_n) = \{\lambda\}.
\end{eqnarray}
\end{enumerate}
\end{definition}

Note that in (\ref{eq:SuffQSep}) the representation of $\{\lambda\}$ by the elements $\theta_1$ to $\theta_n$ is \textit{unique}, since $\theta_i \not= \theta'_i$ implies $\tau_i(\theta_i) \cap \tau_i(\theta'_i) = \emptyset$.

If we define $\Theta \leq \Lambda$ for frames in $\mathcal{F}$ if there is a refining from $\Theta$ to $\Lambda$, that is, if $\Lambda$ is a refinement of $\Theta$, then $(\mathcal{F};\leq)$ becomes a join-semilattice, where $\Theta \vee \Lambda$ is the minimal common refinement of $\Theta$ and $\Lambda$. We may always add an absolutely coarsest frame $\mathcal{E} = \{e\}$ such that for any $\Theta \in \mathcal{F}$ there is a refining $\tau(e) = \Theta$ from $\mathcal{E}$ to $\Theta$ The augmented system $\mathcal{F} \cup \{\mathcal{E}\}$ is still a f.c.f. and $\mathcal{E}$ is the bottom element $\mathcal{E} \leq \Theta$ for all frames. 

If $\Theta$ and $\Lambda$ are two frames of a f.c.f, then there are refinings $\tau_1$ and $\tau_2$ from $\Theta$ and $\Lambda$ to the minimal common refinement $\Theta \vee \Lambda$. A pair of elements $\theta$ and $\lambda$ of frames $\Theta$ and $\Lambda$ are called compatible, if $\tau_1(\theta) \cap \tau_2(\lambda) \not = \emptyset$. These elements represent jointly possible answers relative to two questions. For any subset $S$ of $\Theta$,  
\begin{eqnarray} \label{eq:SetTrOp}
t_\Lambda(S) = \{\lambda \in \Lambda:\tau_1(\theta) \cap \tau_2(\lambda) \not = \emptyset \textrm{ for some}\ \theta \in S\}
\end{eqnarray}
is the subset of elements of $\Lambda$ which are compatible with some element of $S$. Note that if $\Lambda \leq \Theta$, then $t_\Lambda$ is identical to the saturation map defined in (\ref{eq:DefOfSat}). The map $t_\Lambda$ between the power sets of $\Theta$ and $\Lambda$ is also called a transport operator; see Section \ref{subsec:AlgOfProbPot} for more on this subject.

The two most important examples of f.c.f are join-semilattices of partitions and multivariate models. These two instances of a f.c.f shall be presented briefly here:

\begin{example} \textit{Join-Semilattices of Partitions:}

Let $U$ be any set, called the universe, representing a set of possible worlds. In this frame, questions can be modelled by equivalence relations $\equiv$ on $U$, the idea being that we have $u \equiv u'$ if the question has the same answer in the worlds $u$ and $u'$ respectively. The equivalence classes of such an equivalence relation form a \textit{partition} $P$ of $U$, The equivalence classes are the blocks of the partition. Therefore we may consider any block of the partition as a possible answer to the question. In this perspective, we may consider the set $\Theta_P$ of blocks of $P$ whose elements represent possible answers.

Now we consider a family of questions $D$. Any question $x \in D$ is thought to be described either by an equivalence relation $\equiv_x$ or equivalently by the associated partition $P_x$. A question $x$ will be considered as finer as a question $y$, if $u \equiv_x u'$ implies $u \equiv_y u'$, or equivalently, if every block of $P_x$ is contained in a (unique) block of $P_y$. We then write $P_y \leq P_x$. This defines a partial order between partitions of $U$, it is in fact the opposite order of $(part(U),\leq)^\vartheta$ usually considered in the literature \cite{graetzer78}. Under this (opposite) order, the \textit{join} or \textit{supremum} of two partitions $P_1$ and $P_2$, written as $P_1 \vee P_2$, is given by the partition whose blocks are the nonempty intersections $B_1 \cap B_2$ of blocks $B_1$ of $P_1$ and $B_2$ of $P_2$. The set of all partitions together with this (or the opposite) order, $(part(U),\leq)$, is in fact a lattice \cite{graetzer78}. That is, there exist also a meet or infimum between any finite family of partitions.

We now assume that the family of partitions $P_x$ for $x \in D$ forms a (sub)-join-semilattice of $(part(U),\leq)$. Let $\mathcal{F}$ be the family of frames $\Theta_{P_x}$ for $x \in D$. If $P_y \leq P_x$, then consider the map $\tau_{y,x} : \Theta_{P_y} \rightarrow 2^{\Theta_{P_y}}$ defined by
\begin{eqnarray*}
\tau_{y,x}(B_y) = \{B_x:B_x \subseteq B_y\}
\end{eqnarray*}
if $B_x$ and $B_y$ denote blocks of partitions $P_x$ and $P_y$ respectively. Obviously $\tau_{y,x}$ is a refining of $\Theta_{P_y}$, $\Theta_{P_x}$ a refinement of $\Theta_{P_x}$ and $\Theta_{P_y}$ a coarsening of $\Theta_{P_x}$. It can easily be verified that the family $(\mathcal{F},\mathcal{R})$, where $\mathcal{R}$ is the set of all maps $\tau_{y,x}$ for partitions $P_y \leq P_x$, is a family of compatible frames. In particular, the lattice of all partitions of $U$, $(part(U),\leq)$ induces in this way a f.c.f. Note that the top partition of $U$ into singleton sets $\{u\}$ for $u \in U$ is a refining for for any frame $\Theta_P$. Similarly, the bottom partition, consisting of the single block $U$, is a coarsening of all frames $\Theta_P$.
\end{example}

\begin{example} \textit{Multivariate Models:}
This very popular model considers a countable set of variables $X_i$ for $i = 1,2 \ldots$, where each variable $X_i$ has a domain of possible values $\Theta_i$, which we assume here to be finite. Then subsets of variables $X_i$ with $i \in s \subseteq \{1,2,\ldots\}$ are considered, whose domains are given by
\begin{eqnarray*}
\Theta_s = \prod_{i \in s} \Theta_i.
\end{eqnarray*}
Often only finite subsets $s$ of variables are considered. These frames or domains $\Theta_i$ can also be seen as partitions of the universe
\begin{eqnarray*}
U = \prod_i \Theta_i.
\end{eqnarray*}
The refining maps $\tau_{s.t}$, where $s \subseteq t$ is given by the inverse of the projection $\pi(x_t) = x_t \vert s$, where $x_t$ denotes a tuple from $\Theta_t$ and $x_t \vert s$ the restriction of this tuple to the subset $s$. So, the family of domains $\Theta_s$ for $s \subseteq \{1,2,\ldots\}$, together with the refining maps from $\Theta_s$ to $\Theta_t$, where $s \subseteq t$, form a f.c.f.  In this case we have $\Theta_s \leq \Theta_t$ in the order between frames if and only if $s \subseteq t$. In this model the partial order $(\mathcal{F},\leq)$ defines a \textit{distributive lattice}, which is isomorphic to a subset lattice of the set of variables.. Such multivariate models arise for instance in probabilistic argumentation based on propositional logic, where the variables are binary \cite{haennikohlas00}. Another important case of multivariate models uses continuous real-valued variables, where $\Theta_s$ becomes $\mathbb{R}^s$, see Section \ref{subsec:AbsContPAS} below.
\end{example}

The concept of a f.c.f has been introduced in \cite{shafer76} in a similar way. In \cite{shafer76} additional conditions are required, in particular, that any frame has refinings in the family, excluding thus an ultimate refining. This eliminates f.c.fs related to lattices of partitions. On the other hand, a f.c.f does not need to include an ultimate refining, it is thus slightly more general than the f.c.f obtained from the join-subsemilattices of partitions. For another discussion of f.c.f see \cite{cuzzolin05}.

\subsection{Conditional Independence} \label{subsec:CondInd}

An important concept in a f.c.f is the one of \textit{conditional independence} between frames. Consider a finite collection of frames $\Theta_i$, $i = 1,\ldots,n$, from a f.c.f $(\mathcal{F},\mathcal{R})$, and let $\tau_i$ denote the refinings from $\Theta_i$ to the join (or common minimal refinement) $\Theta_1 \vee \ldots \vee \Theta_n$. What are the mutually compatible elements from these $n$ frames? They are collected in the following set of tuples 
\begin{eqnarray*}
R(\Theta_1,\ldots,\Theta_n) = \{(\theta_1,\ldots,\theta_n):\theta_i \in \Theta_i,\cap_{i=1}^n \tau_i(\theta_i) \not= \emptyset\}.
\end{eqnarray*}
Note that $\cap_{i=1}^n \tau_i(\theta_i) \not= \emptyset$ implies $\cap_{i=1}^n \tau_i(\theta_i) = \{\theta\}$ for some elment $\theta$ of the minimal common refinement $\Theta_1 \vee \ldots \vee \Theta_n$. The frames $\Theta_1$ to $\Theta_n$ are called mutually \textit{independent}, if
\begin{eqnarray*}
R(\Theta_1,\ldots,\Theta_n) = \Theta_1 \times \cdots \times \Theta_n.
\end{eqnarray*}
Fix an element $\lambda$ in some other frame $\Lambda$. The elements of $\Theta_i$ which are compatible among themselves as well as with $\lambda$ are given by
\begin{eqnarray*}
R_\lambda(\Theta_1,\ldots,\Theta_n) =  \{(\theta_1,\ldots,\theta_n):((\theta_1,\ldots,\theta_n,\lambda) \in R(\Theta_1,\ldots,\Theta_n,\Lambda)\}.
\end{eqnarray*}
Note that here $\Lambda$ is not necessarily different from every $\Theta_i$. The collection of frames $\Theta_1,\ldots,\Theta_n$ is called \textit{conditionally independent} given $\Lambda$, if for all $\lambda \in \Lambda$ we have
\begin{eqnarray*}
R_\lambda(\Theta_1,\ldots,\Theta_n) = R_\lambda(\Theta_1) \times \cdots \times R_\lambda(\Theta_n).
\end{eqnarray*}
Then we write $\bot \{\Theta_1,\ldots,\Theta_n\} \vert \Lambda$ or, for $n = 2$ also $\Theta_1 \bot \Theta_2 \vert \Lambda$. Conditional independence means that once $\lambda$ is given (as an answer to $\Lambda$), then knowing any $\theta_i$ (as an answer to $\Theta_i$ compatible with $\lambda$) does not restrict the possible $\theta_j$ (an answers to $\Theta_j$) if $i \not= j$. This relation has been studied in \cite{kohlasmonney95,kohlas17}. It has been shown there that it satisfies the properties given in the following theorem.

\begin{theorem} \label{th:Q-Separoid}
The relation $\Theta_1 \bot \Theta_2 \vert \Lambda$ in a f.c.f $(\mathcal{F},\mathcal{R})$ satisfies
\begin{description}
\item [C1] $\Theta \bot \Lambda \vert \Lambda$ for all $\Theta,\Lambda \in \mathcal{F}$,
\item [C2] $\Theta_1 \bot \Theta_2 \vert \Lambda$ implies $\Theta_2 \bot \Theta_1 \vert \Lambda$,
\item [C3] $\Theta_1 \bot \Theta_2 \vert \Lambda$ and $\Theta \leq \Theta_2$ imply $\Theta_1 \bot \Theta \vert \Lambda$,
\item [C4] $\Theta_1 \bot \Theta_2 \vert \Lambda$ implies $\Theta_1 \bot \Theta_2 \vee \Lambda \vert \Lambda$.
\end{description}
\end{theorem}

If $(\mathcal{F};\leq)$ is a \textit{distributive lattice}, as for instance in the case of a multivariate model, then a few more properties hold for the relation of conditional independence \cite{kohlas17}.

\begin{theorem} \label{th:Q-Separoid}
If $(\mathcal{F};\leq)$ is a distributive lattice, the relation $\Theta_1 \bot \Theta_2 \vert \Lambda$ in a f.c.f $(\mathcal{F},\mathcal{R})$ satisfies
\begin{description}
\item [C5] $\Theta_1 \bot \Theta_2 \vert \Lambda$ and $\Theta \leq \Theta_2$ imply $\Theta_1 \bot \Theta_2 \vert \Lambda \vee \Theta$,
\item [C6] $\Theta_1 \bot \Theta_2 \vert \Lambda$ and $\Theta_1 \bot \Theta \vert \Theta_2 \vee \Lambda$ imply $\Theta_1 \bot \Theta_2 \vee \Theta \vert \Lambda$,
\item [C7] If  $\Lambda \leq \Theta_2$ and $\Theta \leq \Theta_2$, then $\Theta_1 \bot \Theta_2 \vert \Lambda$ and $\Theta_1 \bot \Theta_2 \vert \Theta$ imply $\Theta_1 \bot \Theta_2 \vert \Lambda \wedge \Theta$.
\end{description}
\end{theorem}

A three-place relation, which satisfies conditions C1 up to C7 is called a \textit{strong separoid} \cite{Dawid2001}. Therefore, we call a relation which satisfies only C1 up to C4 a \textit{quasi-separoid} (q-separoid). Note  that C4 is a consequence of C5 and C6, see \cite{kohlas17}

For the relation $\bot \{\Theta_1,\ldots,\Theta_n\} \vert \Lambda$ Theorem \ref{th:Q-Separoid} extends as follows:

\begin{theorem} \label{th:ExtQSeparoid}
Assume $\bot \{\Theta_1,\ldots,\Theta_n\} \vert \Lambda$. Then
\begin{enumerate}
\item If $\sigma$ is a permutation of $\{1,\ldots,n\}$, then $\bot \{\Theta_{\sigma(1)},\ldots,\Theta_{\sigma(n)}\} \vert \Lambda$,
\item if $J \subseteq \{1,\ldots,n\}$, then $\bot \{\Theta_j:j \in J\} \vert \Lambda$,
\item if $\Theta \leq \Theta_1$, then $\bot \{\Theta,\Theta_2,\ldots,\Theta_n\} \vert \Lambda$,
\item $\bot \{\Theta_1 \vee \Theta_2,\Theta_3,\ldots,\Theta_n\} \vert \Lambda$,
\item $\bot \{\Theta_1 \vee \Lambda,\Theta_2,\ldots,\Theta_n\} \vert \Lambda$.
\end{enumerate}
\end{theorem}
For a proof we refer again to \cite{kohlas17}.

In the case of a multivariate model we have $\Theta_s \bot \Theta_t \vert \Theta_r$ if and only $s \cap t \subseteq r$. 

If $(\mathcal{F};\leq)$ is a lattice, then $\Theta \bot \Lambda \vert \Theta \wedge \Lambda$ may hold for all pair of frames, as for instance in a multivariate model. This does not hold in general, for example in join-semilattices of partitions. So, what does this mean? To answer this question, let $\mu_1$ and $\mu_2$ the refinings of $\Theta \wedge \Lambda$ to $\Theta$ and $\Lambda$ respectively, and $\tau_1$ and $\tau_2$ the refinings of $\Theta$ and $\Lambda$ to $\Theta \vee \Lambda$. Consider elements $\theta \in \Theta$, $\lambda \in \Lambda$ and $\chi \in \Theta \wedge \Lambda$. Note that $\tau_1(\theta) \cap \tau_1(\mu_1(\chi)) \not= \emptyset$ if and only if $\theta \in \mu_1(\chi)$, and, similarly, $\tau_2(\lambda) \cap \tau_1(\mu_1(\chi)) \not= \emptyset$ if and only if $\lambda \in \mu_2(\chi)$. The conditional independence condition $\Theta \bot \Lambda \vert \Theta \wedge \Lambda$ implies
\begin{eqnarray*}
\tau_1(\theta) \cap \tau_2(\lambda) \cap \tau_1(\mu_1(\chi)) = \tau_1(\theta) \cap \tau_2(\lambda) \cap \tau_2(\mu_2(\chi)) \not= \emptyset
\end{eqnarray*}
if $\theta \in \mu_1(\chi)$ and $\lambda \in \mu_2(\chi)$. Therefore, if $\theta \in \mu_1(\chi)$ and $\lambda \in \mu_2(\chi)$, then $\tau_1(\theta) \cap \tau_2(\lambda) \not= \emptyset$. This can be expressed in the following way: If $\Theta$, $\Lambda$ and $\Theta \wedge \Lambda$ are considered as partitions of $\Theta \vee \Lambda$, then, if $\theta$ and $\lambda$ are in the same block of $\Theta \wedge \Lambda$, there is an element $\zeta \in \Theta \vee \Lambda$ such that $\theta$ and $\zeta$ are in the same block of $\Theta$ and $\lambda$ and $\zeta$ are in the same block of $\Lambda$. Note that sublatticies of a partition lattice satisfying this condition for any pair of blocks are also called partition lattices of type I \cite{graetzer78}. So, $\Theta \bot \Lambda \vert \Theta \wedge \Lambda$ holds only in very special types of partition lattices. Nevertheless this special case is important as the case of the multivariate models shows.

If $S$ is a subset of some frame $\Theta$ of a f.c.f $(\mathcal{F},\mathcal{R})$ and $\Lambda$ any other frame of the f.c.f, then we have
\begin{eqnarray*}
t_\Lambda(S) = \bigcup_{\theta \in S} R_\theta(\Lambda)
\end{eqnarray*}
for the set of all elements $\lambda$ of frame $\Lambda$, compatible with some element $\theta$ of the subset $S$ of frame $\Theta$ (see (\ref{eq:SetTrOp}). For any element $\theta \in \Theta$ we write $t_\Lambda(\theta)$ instead of $t_\Lambda(\{\theta\})$. This is the transport of an element of frame $\Theta$ to the frame $\Lambda$. Note that $t_\Lambda(\theta)$ is a set, a subset of $\Lambda$. If $\Lambda \leq \Theta$, then $t_\Lambda(\theta)$ is a one-element set. Further, if $\theta \in \Theta$ and $\lambda \in \Lambda$, we write $\theta \sim \lambda$ if the two elements are compatible, that is $\lambda \in R_\theta(\Lambda)$ 
or, equivalently $\theta \in R_\lambda(\Theta)$.

Let us add some results on conditional independent frames, which we need later (see Section \ref{sec:MostProb}). The first result states that if frames $\Theta$ and $\Lambda$ are conditionally independent given a frame $\Lambda_1$, then, if $\lambda \in \Lambda$ is compatible with $\lambda_1 \in \Lambda_1$, then any element $\theta \in \Theta$, compatible with $\lambda_1$ is also compatible with $\lambda$.

\begin{lemma} \label{CondIndepRes1}
Let $(\mathcal{F},\mathcal{R})$ be an f.c.f and assume $\Theta \bot \Lambda \vert \Lambda_1$ for $\Theta,\Lambda, \Lambda_1 \in \mathcal{F}$. Then  if $\lambda \in \Lambda$, and $\lambda_1 \in \Lambda_1$, $\lambda \sim \lambda_1$ implies $R_{\lambda_1}(\Theta) \subseteq R_\lambda(\Theta)$ or $\theta \sim \lambda_1 \Rightarrow \theta \sim \lambda$..
\end{lemma}

\begin{proof}
Recall that $\Theta \bot \Lambda \vert \Lambda_1$ means that $R_{\lambda_1}(\Theta,\Lambda) = R_{\lambda_1}(\Theta) \times R_{\lambda_1}(\Lambda)$. Le'ts first translate this statement in a different form, useful for the proof. Let
\begin{enumerate}
\item $\tau_1$, $\mu_1$ be refinings of $\Theta$ and $\Lambda$ to $\Theta \vee \Lambda$,
\item $\tau_2$, $\nu_1$ be refinings of $\Theta$ and $\Lambda_1$ to $\Theta \vee \Lambda_1$,
\item $\mu_2$, $\nu_2$ be refinings of $\Lambda$ and $\Lambda_1$ to $\Lambda \vee \Lambda_1$,
\item $\epsilon_1$, $\epsilon_2$, $\epsilon_3$ refinings of $\Theta \vee \Lambda$, $\Theta \vee \Lambda_1$ and $\Lambda \vee \Lambda_1$ to $\Theta \vee \Lambda \vee \Lambda_1$.
\end{enumerate}
Now, by definition, a pair $(\theta,\lambda)$ belongs to $R_{\lambda_1}(\Theta,\Lambda)$, if the triple of elements $(\theta,\lambda,\lambda_1)$ belongs to $R(\Theta,\Lambda,\Lambda_1)$ and this in turn is the case if the intersection of the refinings of $\theta$, $\lambda$ and $\lambda_1$ to $\Theta \vee \Lambda \vee \Lambda_1$ is not empty. Now, this can be expressed in different ways using the refinings defined above:
\begin{eqnarray} \label{eq:TripleComp}
\emptyset &\not=& \epsilon_1(\tau_1(\theta) \cap \mu_1(\lambda)) \cap \epsilon_2(\nu_1(\lambda_1)) \nonumber \\
&=& \epsilon_1(\tau_1(\theta)) \cap \epsilon_3(\mu_2(\lambda) \cap \nu_2(\lambda_1)) \nonumber \\ 
&=& \epsilon_2(\tau_2(\theta) \cap \nu_1(\lambda_1)) \cap \epsilon_1(\mu_1(\lambda)).
\end{eqnarray}
This implies
\begin{eqnarray} \label{eq:PairComp}
\tau_2(\theta) \cap \nu_1(\lambda_1) \not= \emptyset, \quad \mu_2(\lambda) \cap \nu_2(\lambda_1) \not= \emptyset
\end{eqnarray}
or, in other words, $\theta \in R_{\lambda_1}(\Theta)$ and $\lambda \in R_{\lambda_1}(\Lambda)$ (or $\lambda \sim \lambda_1$). If $\Theta \bot \Lambda \vert \Lambda_1$, then (\ref{eq:PairComp}) implies also (\ref{eq:TripleComp}). We exploit this now for the proof of the lemma.

So assume $\lambda \sim \lambda_1$, that is $\mu_2(\lambda) \cap \nu_2(\lambda_1) \not= \emptyset$ and $\theta \in R_{\lambda_1}(\Theta)$, that is $\tau_2(\theta) \cap \nu_1(\lambda_1) \not= \emptyset$. Then since $\Theta \bot \Lambda \vert \Lambda_1$, (\ref{eq:TripleComp}) holds, which implies $\tau_1(\theta) \cap \mu_1(\lambda)$, hence $\theta \in R_\lambda(\Theta)$. 
\end{proof}

Next, we assure that if elements $\theta \in \Theta$ and $\lambda \in \Lambda$ are compatible, then there is an element $\lambda_1 \in \Lambda_1$ such that the triple of elements $(\theta,\lambda,\lambda_1)$ is compatible and so are the pairs $(\theta,\lambda_1)$ and $(\lambda,\lambda_1)$.

\begin{lemma} \label{CondIndepRes2}
Let $\theta \sim \lambda$. Then there is an element $\lambda_1 \in \Lambda_1$ such that $(\theta,\lambda,\lambda_1) \in R(\Theta,\Lambda,\Lambda_1)$ and $\theta \sim \lambda_1$ and $\lambda \sim \lambda_1$.
\end{lemma}

\begin{proof}
We use the refinings defined in the proof of the previous lemma. Then $\theta \sim \lambda$ means that $\tau_1(\theta) \cap \mu_1(\lambda)$ is not empty and so is $\epsilon_1(\tau_1(\theta) \cap \mu_1(\lambda))$ as a subset of $\Theta \vee \Lambda \vee \Lambda_1$. Now, $\epsilon_2 \circ \nu_1$ is the refining of frame $\Lambda_1$ to $\Theta \vee \Lambda \vee \Lambda_1$. This refining of $\Lambda_1$ to $\Theta \vee \Lambda \vee \Lambda_1$ covers the latter frame. Therefore there must be a $\lambda_1 \in \Lambda_1$ such that $\epsilon_1(\tau_1(\theta) \cap \mu_1(\lambda)) \cap \epsilon_2(\nu_1(\lambda_1)) \not= \emptyset$. But this means that $(\theta,\lambda,\lambda_1) \in R(\Theta,\Lambda,\Lambda_1)$. The rest follows then from (\ref{eq:TripleComp}).
\end{proof}

The next lemma states further results on compatibility of elements on different frames.

\begin{lemma} \label{CondIndepRes3}
Let $(\mathcal{F},\mathcal{R})$ be an f.c.f and assume $\Theta_1 \bot \Theta_2 \vert \Lambda$ for $\Theta_1,\Theta_2, \Lambda \in \mathcal{F}$ and let $\tau_1$ and $\tau_2$ respectively denote the refinings of $\Theta_1, \Theta_2$ to $\Theta_1 \vee \Theta_2$.
\begin{enumerate}
\item If $(\theta_1,\theta_2) \in R_\lambda(\Theta_1,\Theta_2)$ for $\lambda \in \Lambda$, then $\emptyset \not= \tau_1(\theta_1) \cap \tau_2(\theta_2) = \{\theta\}$ where $\theta \in R_\lambda(\Theta_1 \vee \Theta_2)$.
\item If $\theta \in R_\lambda(\Theta_1 \vee \Theta_2)$, then $t_{\Theta_1}(\theta) \in R_\lambda(\Theta_1)$ and $t_{\Theta_2}(\theta) \in R_\lambda(\Theta_2)$ and further $\tau_1(t_{\Theta_1}(\theta)) \cap \tau_2(t_{\Theta_2}(\theta)) = \{\theta\}$. 
\item The map $\theta \inÊR_\lambda(\Theta_1 \vee \Theta_2) \mapsto (t_{\Theta_1}(\theta)), t_{\Theta_2}(\theta)) \in R_\lambda(\Theta_1,\Theta_2)$ establishes a bijection between $R_\lambda(\Theta_1 \vee \Theta_2)$ and $R_\lambda(\Theta_1,\Theta_2)$.
\end{enumerate}
\end{lemma}

\begin{proof}
Consider the following refinings:
\begin{enumerate}
\item $\tau_1$ and $\tau_2$ refinings from $\Theta_1$ and $\Theta_2$ to $\Theta_1 \vee \Theta_2$,
\item $\tau$ and $\mu$ refinings from $\Theta_1 \vee \Theta_2$ and $\Lambda$ to $\Theta_1 \vee \Theta_2 \vee \Lambda$,
\item $\tau'_1$ and $\tau'_2$ the refinings of $\Theta_1$ and $\Theta_2$ to $\Theta_1 \vee \Lambda$ and $\Theta_2 \vee \Lambda$,
\item $\nu_1$ and $\nu_2$ the refinings of $\Lambda$ to $\Theta_1 \vee \Lambda$ and $\Theta_2 \vee \Lambda$,
\item $\mu_1$ and $\mu_2$ the refinings from $\Theta_1 \vee \Lambda$ and $\Theta_2 \vee \Lambda$ to $\Theta_1 \vee \Theta_2 \vee \Lambda$.
\end{enumerate}
Then, $(\theta_1,\theta_2) \in R_\lambda(\Theta_1,\Theta_2)$ means that $\tau(\tau_1(\theta_1) \cap \tau_2(\theta_2)) \cap \mu(\lambda) \not= \emptyset$, and therefore $\tau_1(\theta_1) \cap \tau_2(\theta_2) \not= \emptyset$. By the property of minimal common refinements, we have then $\tau_1(\theta_1) \cap \tau_2(\theta_2) = \{\theta\}$ for some element $\theta \in \Theta_1 \vee \Theta_2$. This proves item 1.

Next, $\theta \in R_\lambda(\Theta_1 \vee \Theta_2)$ means that $\tau(\theta) \cap \mu(\lambda) \not= \emptyset$. But, if $\theta_1 = t_{\Theta_1}(\theta)$, then $\tau(\theta) \subseteq \tau(\tau_1(\theta_1))$, hence $\tau(\tau_1(\theta_1)) \cap \mu(\lambda) \not= \emptyset$. Then we have $\tau(\tau_1(\theta_1)) = \mu_1(\tau'_1(\theta_1))$, hence
\begin{eqnarray*}
\mu_1(\tau'_1(\theta_1) \cap \nu_1(\lambda)) = \mu_1(\tau'_1(\theta_1)) \cap \mu_1(\nu_1(\lambda)) = \mu_1(\tau'_1(\theta_1)) \cap \mu(\lambda) \not= \emptyset.
\end{eqnarray*}
From this we conclude that $\tau'_1(\theta_1) \cap \nu_1(\lambda) \not= \emptyset$, which means that $\theta_1 = t_{\Theta_1}(\theta) \in R_\lambda(\Theta_1)$. And $t_{\Theta_2}(\theta) \in R_\lambda(\Theta_2)$ is proved in the same way. Further, $\tau_1(t_{\Theta_1}(\theta)) \cap \tau_2(t_{\Theta_2}(\theta)) = \{\theta\}$ follows since $\theta \in \tau_1(t_{\Theta_1}(\theta))$ and $\theta \in \tau_2(t_{\Theta_2}(\theta))$ and by the property of the minimal common refinement. This shows that item 2 holds.
 
Finally, by item 2 $(t_{\Theta_1}(\theta)),t_{\Theta_2}(\theta))$ belongs to $R_\lambda(\Theta_1) \times R_\lambda(\Theta_2) = R_\lambda(\Theta_1,\Theta_2)$ if $\theta \in R_\lambda(\Theta_1 \vee \Theta_2)$. The map $\theta \inÊR_\lambda(\Theta_1 \vee \Theta_2) \mapsto (t_{\Theta_1}(\theta)),t_{\Theta_2}(\theta)) \in R_\lambda(\Theta_1,\Theta_2)$ is invertible, since $\{\theta\} = \tau_1(t_{\Theta_1}(\theta)) \cap \tau_2(t_{\Theta_2}(\theta))s$ and it is onto, since $(\theta_1,\theta_2) \in R_\lambda(\Theta_1,\Theta_2)$ implies by item 1 that $\theta \in R_\lambda(\Theta_1 \vee \Theta)$ if $\{\theta\} = \tau_1(\theta_1) \cap \tau_2(\theta_2)$.
\end{proof}

This concludes the discussion of conditional independence in f.c.f.


\section{Reasoning with PAS} \label{sec:ReasPAS}

\subsection{Independent PAS} \label{subsec:indepPAS}

In this section, we fix a family of compatible frames (f.c.f) $(\mathcal{F},\mathcal{R})$ and consider PAS defined on frames $\Theta \in \mathcal{F}$. We extend the combination procedure of two (or more) PAS, as described in Section \ref{sec:PAS} for PAS on the same frame, to PAS on different frames in the f.c.f. So, consider two PAS $(\Omega_1,p_1,X_1,\Theta_1)$ and $(\Omega_2,p_2,X_2,\Theta_2)$ for two frames $\Theta_1$ and $\Theta_2$ from $\mathcal{F}$. As in section \ref{sec:PAS} we consider combined assumptions $(\omega_1,\omega_2) \in \Omega_1 \times \Omega_2$. But now $X_1(\omega_1)$ is in $\Theta_1$, whereas $X_2(\omega_2)$ lies in $\Theta_2$. In order to combine or compare the two implications $X_1(\omega_1)$ and $X_2(\omega_2)$ of the two assumptions $\omega_1$ and $\omega_2$ we consider the elements in $\Theta_1 \vee \Theta_2$ compatible respectively with $X_1(\omega_1)$ and $X_2(\omega_2)$. According to Section \ref{sec:FcF} these elements are given by the sets $\tau_1(X_1(\omega_1))$ and $\tau_2(X_2(\omega_2))$, where $\tau_1$ and $\tau_2$ are  the refinings of $\Theta_1$ and $\Theta_2$  to their common refinement $\Theta_1 \vee \Theta_2$. So, the elements in $\Theta_1 \vee \Theta_2$ compatible both with $X_1(\omega_1)$ and $X_2(\omega_2)$ are in the intersection $\tau_1(X_1(\omega_1)) \cap \tau_2(X_2(\omega_2))$. According to the Existence of a Minimal Common Refinement in the definition of a f.c.f (see Section \ref{sec:FcF}), this intersection is either empty or contains exactly one atom. In the first case, the two assumptions $\omega_1$ and $\omega_2$ are \textit{contradictory}, no element in $\Theta_1 \vee \Theta_2$ is compatible with both assumptions. As argued in Section \ref{sec:PAS} such pairs of assumptions are to be eliminated as impossible, only the remaining pairs are to be accepted.

In order to express the combination rule for the two PAS $(\Omega_1,p_1,X_1,\Theta_1)$ and $(\Omega_2,p_2,X_2,\Theta_2)$ more formally, we simplify notation by writing $t_{\Theta_1 \vee \Theta_2}(\theta_1) = \tau_1(\{\theta_1\})$ and $t_{\Theta_1 \vee \Theta_2}(\theta_2) = \tau_2(\{\theta_2\})$. Then we define the combined PAS of $(\Omega_1,p_1,X_1,\Theta_1)$ and $(\Omega_2,p_2,X_2,\Theta_2)$ by $(\Omega,p,X,\Theta_1 \vee \Theta_2)$ where
\begin{enumerate}
\item $\Omega = \{(\omega_1,\omega_2) \in \Omega_1 \times \Omega_2:t_{\Theta_1 \vee \Theta_2}(X_1(\omega_1)) \cap t_{\Theta_1 \vee \Theta_2}(X_2(\omega_2)) \not= \emptyset\}$,
\item $p(\omega_1,\omega_2) = k^{-1}p_1(\omega_1)p_2(\omega_2)$ for $(\omega_1,\omega_2) \in \Omega$, provided
\begin{eqnarray*}
k = \sum_{(\omega_1,\omega_2) \in \Omega} p_1(\omega_1)p_2(\omega_2) \not= 0,
\end{eqnarray*}
\item $X(\omega_1,\omega_2) = \lambda$, if $t_{\Theta_1 \vee \Theta_2}(X_1(\omega_1)) \cap t_{\Theta_1 \vee \Theta_2}(X_2(\omega_2)) = \{\lambda\}$, $\lambda \in \Theta_1 \vee \Theta_2$, $(\omega_1,\omega_2) \in \Omega$.
\end{enumerate}
If $k = 0$, then the two PAS $(\Omega_1,p_1,X_1,\Theta_1)$ and $(\Omega_2,p_2,X_2,\Theta_2)$ are called \textit{contradictory}, they cannot be combined into a new PAS. This is also called the combination rule of \textit{independent} PAS. 

For a PAS on a frame $\Theta$ in a f.c.f we may define another operation, namely the one of the projection (or coarsening) of the PAS to a coarser frame $\Lambda \leq \Theta$. If $(\Omega,p,X,\Theta)$ is a PAS relative to the frame $\Theta$, and $\tau$ the refining of $\Lambda$ to $\Theta$, then, for an assumption $\omega \in \Omega$, if $X(\omega) \in \tau(\lambda)$, then $\omega$ implies $\lambda \in \Lambda$. Therefore, we may call the PAS $(\Omega,p,Y,\Lambda)$ with $Y(\omega) = t_\Lambda(X(\omega)) = \lambda$, if $X(\omega) \in \tau(\lambda)$ the \textit{projection} or \textit{coarsening} of the PAS $(\Omega,p,X,\Theta)$ to the frame $\Lambda$. 

So, we have two operations among independent PAS on a f.c.f $(\mathcal{F},\mathcal{R})$. This points to a certain algebraic structure of those PAS. This structure however expresses itself more clearly, when we consider the probability distributions on the frames $\Theta$ associated with the PAS. This will be discussed in the next section. Previously, we consider the question, whether a PAS $(\Omega,p,X,\Theta)$ can somehow also be transported to any other frame $\Lambda$ in the f.c.f, not only to $\Lambda \leq \Theta$, as in the projection operation. We may try to do this by considering the elements in $\Lambda$ which are compatible with the element $X(\omega)$ in $\Theta$. So, we might assign $t_\Lambda(X(\omega)$ in $\Lambda$ to the assumption $\omega$. The point is however, that $t_\Lambda(X(\omega)$ is a \textit{set} in general and not a single element. So, the resulting structure is no more a PAS in the strict sense of Section \ref{sec:PAS}. It can however definitely make sense to consider structures where assumptions imply a \textit{subset} of a frame rather than a single element. This has been extensively described in \cite{kohlasmonney95} in the theory of hints. And this point of view will also be of some help in subsequent sections. 

We now enlarge therefore the point of view a bit in the direction indicated above. We consider PAS $(\Omega,p,X,\Theta)$ relative to a frame $\Theta$ of an f.c.f where, however, $X$ maps $\Omega$ now into the \textit{power set} of $\Theta$, that is, $X(\omega)$ is a non-empty \textit{subset} of $\Theta$. Associated with such a generalized PAS is a \textit{basic probability assignment} (bpa) for subsets $S$ of $\Theta$, defined by
\begin{eqnarray*}
m(S) = \sum_{\omega: X(\omega) = S} p(\omega).
\end{eqnarray*}
If there is no assumption $\omega$ such that $X(\omega) = S$, then put $m(S) = 0$. Obvioulsly, we have
\begin{eqnarray*}
m(S) \geq 0 \textit{ for all subsets}\ S \subseteq \Theta, \quad \sum_{S \subseteq \Theta} m(S) = 1.
\end{eqnarray*}
Further, we have $m(\emptyset) = 0$. 

Using the transport operators $t_{\Theta_1 \vee \Theta_2}(S_1)$ for $S_1 \subseteq \Theta_1$ and $t_{\Theta_1 \vee \Theta_2}(S_2)$ for $S_2 \subseteq \Theta_2$ we can define combination $(\Omega,p,X.\Lambda)$ of two generalized PAS $(\Omega_1,p_1,X_1,\Theta_1)$ and $(\Omega_2,p_2,X_2,\Theta_2)$ along the same lines as the combination of ordinary PAS in Section \ref{sec:PAS}:
\begin{enumerate}
\item $\Lambda = \Theta_1 \vee \Theta_2$,
\item $\Omega = \{(\omega_1,\omega_2) \in \Omega_1 \times \Omega_2):t_{\Theta_1 \vee \Theta_2}(X_1(\omega_1)) \cap t_{\Theta_1 \vee \Theta_2}(X_2(\omega_2)) \not=Ê\emptyset\}$,
\item $p(\omega_1,\omega_2) = k^{-1}p_1(\omega_1)p_2(\omega_2)$, where $k = \sum_{(\omega_1,\omega_2) \in \Omega} p_1(\omega_1)p_2(\omega_2)$ for $(\omega_1,\omega_2) \in \Omega$,
\item $X(\omega_1,\omega_2) = t_{\Theta_1 \vee \Theta_2}(X_1(\omega_1)) \cap t_{\Theta_1 \vee \Theta_2}(X_2(\omega_2))$ for $(\omega_1,\omega_2) \in \Omega$.
\end{enumerate}
Here we assume that $k \not= \emptyset$, otherwise the two PAS are contradictory. This operation is reflected by a corresponding operation between the associated bpa $m_1$ and $m_2$ of the two PAS: For any subset $S$ of $\Theta_1 \vee \Theta_2$, we have for the bpa of the combined PAS,
\begin{eqnarray*}
m(S) = k^{-1} \sum_{t_{\Theta_1 \vee \Theta_2}(S_1) \cap t_{\Theta_1 \vee \Theta_2}(S_2) = S} m_1(S_1)m_2(S_2),
\end{eqnarray*}
with
\begin{eqnarray*}
k =  \sum_{t_{\Theta_1 \vee \Theta_2}(S_1) \cap t_{\Theta_1 \vee \Theta_2}(S_2) \not= \emptyset} m_1(S_1)m_2(S_2).
\end{eqnarray*}
This is known as \textit{Dempster's rule} in Dempster-Shafer theory of evidence, at least if $\Theta_1 = \Theta_2$, \cite{shafer76}, and the present discussion shows how Dempster-Shafer theory is related to (generalized) PAS. 

Similarly, we define transport of a PAS $(\Omega,p,X,\Theta)$ to some frame $\Lambda$ by $(\Omega,p,Y,\Lambda)$, where
\begin{eqnarray*}
Y(\omega) = t_\Lambda(X(\omega)),
\end{eqnarray*}
The bpa of the PAS $(\Omega,p,Y,\Lambda)$ is given for $S \subseteq \Lambda$ by
\begin{eqnarray*}
m(S) = \sum_{t_\Lambda(T)) = S} m_1(T)
\end{eqnarray*}
in terms of the bpa $m_1$ of the PAS $(\Omega,p,X,\Theta)$.

The PAS considered originally in this paper are essentially identical to generalized PAS, where $X(\omega)$ are one-element sets, $X(\omega) =\{\theta\}$ for all $\omega \in \Omega$. In this sense generalized PAS are an extension of \textit{precise} PAS as discussed above. Another important special case of generalized PAS arises, if $X(\omega) = S \subseteq \Theta$ for all $\omega \in \Omega$. This is called a \textit{deterministic} PAS, since it fixes a constant subset $S$ of $\Theta$ for all possible assumptions. 

As before, for any subset $S$ of $\Theta$, we may ask to what degree the hypothesis that the unknown element of $\Theta$ is in $S$ is supported by a generalized PAS $(\Omega,p,X,\Theta)$. And similarly as before, we consider the set of assumptions $\omega$ for which $X(\omega)$ implies $S$, that is
\begin{eqnarray*}
s(S) = \{\omega \in \Omega: X(\omega) \subseteq S\}.
\end{eqnarray*}
Further, we may obtain the probability of this set as
\begin{eqnarray*}
sp(S) = P(s(S)) = \sum_{\omega \in s(S)} p(\omega) = \sum_{T \subseteq S} m(T).
\end{eqnarray*}
This is called the \textit{support function} of the PAS. It corresponds to belief functions in Demster-Shafer theory of evidence. We may also ask to what extend the hypothesis $S$ is not excluded by the PAS, that is
\begin{eqnarray*}
pl(S) = 1 - sp(S^c) = \sum_{T : T \cap S \not= \emptyset} m(T).
\end{eqnarray*}
This is called the \textit{plausibility function} of the PAS. These are well-known functions in Dempster-Shafer theory of evidence, see \cite{shafer76} for more details. For our purposes, the plausibility of singleton sets $S = \{\theta\}$, 
\begin{eqnarray*}
pl(\theta) = \sum_{\omega: \theta \in X(\omega)} p(\omega) = \sum_{T: \theta \in T} m(T)
\end{eqnarray*}
are of particular importance, as we shall see in the next section. This is also called the \textit{likelihood function} of the PAS. Note that if the PAS is precise, then the likelihood function is identical to the probability distribution induced by the PAS.

\subsection{Algebras of Set- and Probability-Potentials} \label{subsec:AlgOfProbPot}

With the operations of combination and projection or transport, the class of PAS, whether precise or generalized, acquire an algebraic flavor. The associated algebraic structures of PAS have been studied elsewhere in detail \cite{kohlas17}. Here we focus on related algebras of probability distributions or bpas associated with probabilistic argumentation structures. 
We have seen in Section \ref{sec:PAS} that any precise PAS $(\Omega,p,X,\Theta)$ induces a discrete probability distribution $sp$ on the frame $\Theta$, whereas generalized PAS induce basic probability assignments. With each concept both combination of (independent) PAS as well as projection or transport respectively can be expressed. So this indicates the existence of associated algebraic structures for probability distributions and bpas. It is well-known that bpas in the multivariate setting form a valuation algebra \cite{shenoyshafer90,kohlas03}. Here this will be extended to the general f.c.f setting and also a related algebraic structure for probability distributions will be presented and discussed. 

Let $(\mathcal{F},\mathcal{R})$ be a f.c.f. We extend the concept of bpa on such a f.c.f by considering $\Psi_\Theta$, the family of all functions $m : \mathcal{P}(\Theta) \rightarrow \mathbb{R}^+ \cup \{0\}$ which assigns all subsets of a frame $\Theta \in \mathcal{F}$, including the empty set, a nonnegative real number. Let
\begin{eqnarray*}
\Psi = \bigcup_{\Theta \in \mathcal{F}} \Psi_\Theta.
\end{eqnarray*}
We call the elements of $\Psi$ (non- normalized) set potentials; non-normalized because the sum of the $m(S)$ equals not necessarily one, and the empty set may have a positive value $m(\emptyset)$. We remark that to any non-null non-normalized bpa $m$ we may associate a uniquely determined normalized $m_n$ bpa in the following way:
\begin{eqnarray} \label{eq:ScaleSetPot}
m^\downarrow(S) = k^{-1}m(S) \textrm{ for}\ S \not= \emptyset, \quad m^\downarrow(\emptyset) = 0,
\end{eqnarray}
where
\begin{eqnarray*}
k = \sum_{S \not= \emptyset} m(S).
\end{eqnarray*}
This process is called normalization or scaling.

We focus first on set potentials and come back later to normalization. Within the family $\Psi$ of set potentials relative to a f.c.f $(\mathcal{F},\mathcal{R})$ we introduce three operations, namely
\begin{enumerate}
\item \textit{Labeling:} $d : \Psi \rightarrow \mathcal{F}$, defined by $m \mapsto d(m) = \Theta$ if $m \in \Psi_\Theta$.
\item \textit{Combination:} $\cdot : \Psi \times \Psi \rightarrow \Psi$, defined by $(m_1,m_2) \mapsto m_1 \cdot m_2$, where, for $S \subseteq d(m_1) \vee d(m_2)$, if $d(m_1) = \Theta_1$ and $d(p_2) = \Theta_2$,
\begin{eqnarray} \label{eq:CombOfSetPot}
 m_1 \cdot m_2(S) = \sum_{S_1 \in \Theta_1, S_2 \in \Theta_2: t_{\Theta_1 \vee \Theta_2}(S_1) \cap t_{\Theta_1 \vee \Theta_2}(S_2) = S} m_1(S_1)m_2(S_2).
\end{eqnarray} \label{eq:TranspOfSetPot}
\item \textit{Transport:} $t : \Psi \times \mathcal{F} \rightarrow \Psi$, defined by $(m,\Theta) \mapsto t_\Theta(m)$, where for $S \subseteq \Theta$,
\begin{eqnarray}
t_\Theta(m)(S) = \sum_{T \in d(m):t_\Theta(T) = S} m(T).
\end{eqnarray}
\end{enumerate}
Here combination is non-normalized, in contrast to the Dempser-Shafer rule of the previous section. The set potential $\mathbf{0}_\Theta(S) = 0$ for all subsets $S$ is the null element of combination on the frame $\Theta$, that is $mÊ\cdot \mathbf{0}_\Theta= \mathbf{0}_\Theta \cdot m = \mathbf{0}_\Theta$ for all $m$ with $d(m) = \Theta$; and the set potential $\mathbf{1}_\Theta(\Theta) = 1$, $\mathbf{1}_\Theta(S) = 0$ if $S \not= \Theta$, is the unit element on $\Theta$, that is $m \cdot \mathbf{1}_\Theta = \mathbf{1}_\Theta = m$ if $d(m) = \Theta$. Note that the combination of two non-null set potentials may well result in the null potential. Then the two set potentials are called contradictory.

The family of set potentials $\Psi$ on a f.c.f $(\mathcal{F},\mathcal{R})$ satisfies the following properties:
\begin{description}
\item[A0] \textit{Quasi-Separoid:} $(\mathcal{F},\leq,\bot)$ is a quasi-separoid.
\item[A1] \textit{Semigroup:} $(\Psi,\cdot)$ is a commutative semigroup.
\item[A2] \textit{Labeling:} $d(m_1 \cdot m_2) = d(m_1) \vee d(m_2)$, $d(t_\Lambda(m)) = \Lambda$.
\item[A3] \textit{Unit and Null:} For all $\Theta \in \mathcal{F}$ there is a unit element $\mathbf{1}_\Theta$ with $d(\mathbf{1}_\Theta) = \Theta$ and a null element $\mathbf{0}_\Theta$ with $d(\mathbf{0}_\Theta) = \Theta$ such that
\begin{enumerate}
\item $m \cdot \mathbf{1}_\Theta = m$ and $mÊ\cdot \mathbf{0}_\Theta = \mathbf{0}_\Theta$ if $d(m) = \Theta$,
\item $t_\Theta(m) = \mathbf{0}_\Theta$ if and only if $m = \mathbf{0}_{d(m)}$,
\item $m \cdot \mathbf{1}_\Theta = t_{d(m) \vee \Theta}(m)$,
\item $\mathbf{1}_\Theta \cdot \mathbf{1}_\Lambda = \mathbf{1}_{\Theta \vee \Lambda}$.
\end{enumerate}
\item[A4] \textit{Transport:} $\Theta_1 \bot \Theta_2 \vert \Lambda$ and $d(m) = \Theta_1$ imply 
\begin{eqnarray*}
t_{\Theta_2}(m) = t_{\Theta_2}(t_\Lambda(m)).
\end{eqnarray*}
\item[A5] \textit{Combination:} $\Theta_1 \bot \Theta_2 \vert \Lambda$ and $d(m_1) = \Theta_1$, $d(m_2) = \Theta_2$ imply
\begin{eqnarray*}
t_\Lambda(m_1 \cdot m_2) = t_\Lambda(m_1) \cdot t_\Lambda(m_2).
\end{eqnarray*}
\item[A6] \textit{Identity:} $d(m) = \Theta$ implies $t_\Theta(m) = m$.
\end{description}
Most of these properties are obvious. The important conditional independence properties A4 and A5 are proved in \cite{kohlasmonney95}. An algebraic structure satisfying these properties is called a \textit{(generalized) information algebra} in \cite{kohlas17} \footnote{In \cite{kohlas03}, only idempotent algebras are called information algebras, here we drop this requirement}. Since combination and transport of PAS are reflected by (non normalized) combination and transport of the associated set potentials, these operations on PAS may be as well executed in the algebra of bpa. As we shall see in the following Section Ê\ref{sec:LocComp} this may have great advantages.

Next we consider the family $\Phi_\Theta$ of functions $p : \Theta \rightarrow \mathbb{R}^+ \cup \{0\}$ of non-negative real-valued functions on the frame $\Theta$ and we define
\begin{eqnarray*}
\Phi = \bigcup_{\Theta \in \mathcal{F}} \Phi_\Theta.
\end{eqnarray*}
Since any such non-null function $p$ can be normalized to a probability distribution over $\Theta$ by
\begin{eqnarray} \label{eq:ScaledProbPot}
p^\downarrow(\theta) = k^{-1} p(\theta) \textrm{ for}\ \theta \in \Theta, 
\end{eqnarray}
where
\begin{eqnarray*}
k = \sum_{\theta \in \Theta} p(\theta),
\end{eqnarray*}
we call the elements of $\Phi$ \textit{probability potentials}, or short \textit{potentials}. Just as a probability distribution is essentially a special bpa, namely one whose probability assignments are different from zero only for one-element sets, such that $p(\theta) = m(\{\theta\})$ is a probability distribution, probability potentials are essentially identical to set potentials $m$ where $m(S) \not= 0$ only if $S$ is a one element set, and where $p(\theta) = m(\{\theta\})$ is a probability potential. Just as for bpas we may define for any set potential $m$,
\begin{eqnarray*}
pl_m(\theta) = \sum_{T:\theta \in T} m(T),
\end{eqnarray*}
the likelihood function of $m$. It is a probability potential. Now, clearly, any probability potential $p$ is the likelihood function of the set potential $m$ with $m(\{\theta\}) = p(\theta)$, hence $p = pl_m$. We are going to exploit this relation between set and probability potentials. In order to facilitate this discussion we we define for any potential $p$ the corresponding set potential $m_p$ by $m_p(\{\theta\}) = p(\theta)$, $m(S) = 0$ for any subset $S$ of cardinality different form one.

We define in $\Phi$ three operations similar to the ones for set potentials, namely
\begin{enumerate}
\item \textit{Labeling:} $d : \Phi \rightarrow \mathcal{F}$, defined by $p \mapsto d(p) = \Theta$ if $p \in \Phi_\Theta$.
\item \textit{Combination:} $\cdot : \Phi \times \Phi \rightarrow \Phi$, defined by $(p_1,p_2) \mapsto p_1 \cdot p_2$, where, for $\theta \in d(m_1) \vee d(m_2)$, if $d(p_1) = \Theta_1$ and $d(p_2) = \Theta_2$,
\begin{eqnarray} \label{eq:CombOfProbPot}
 p_1 \cdot p_2(\theta) = p_1(t_{\Theta_1}(\theta)) \cdot p_2(t_{\Theta_2}(\theta))).
\end{eqnarray} 
\item \textit{Transport:} $\pi : \Phi \times \mathcal{F} \rightarrow \Phi$, defined by $(p,\Lambda) \mapsto \pi_\Lambda(p) = pl_{t_\Lambda(m_p)}$.
\end{enumerate}

The transport operation maps a potential to a potential and can be described more explicitly by
\begin{eqnarray} \label{eq:TranspOfProbPot}
\pi_\Lambda(p)(\lambda) = \sum_{\theta \in \Theta:\lambda \in t_\Lambda(\theta)} p(\theta)
\end{eqnarray}
for any $\lambda \in \Lambda$, if $d(p) = \Theta$. In case $\Lambda \leq \Theta$ this corresponds to projection,
\begin{eqnarray*}
\pi_\Lambda(p)(\lambda) = \sum_{\theta \in \tau(\lambda)} p(\theta)
\end{eqnarray*}
if $\tau$ is the refining of $\Lambda$ to $\Theta$. If $p$ is a probability distribution, then so is $\pi_\Lambda(p)$ in this case.

Now, the map $m \mapsto pl_m$ is a map from set potentials to probability potentials with nice properties, which are important for computing with probability potentials. The main result is the following one:

\begin{theorem} \label{th:CombOfPl}
If $m_1$ and $m_2$ are set potentials in $\Psi$, then in $\Phi$ we have
\begin{eqnarray} \label{eq:CompofLikeliohood}
pl_{m_1 \cdot m_2} = pl_{m_1} \cdot pl_{m_2}.
\end{eqnarray}
\end{theorem}

\begin{proof}
Assume that $m_1$ and $m_2$ are set potentials on domains $\Theta_1$ and $\Theta_2$ and consider an element $\theta \in \Theta_1 \vee \Theta_2$. Then we have
\begin{eqnarray*}
\lefteqn{pl_{m_1 \cdot m_2}(\theta) = \sum_{S:\theta \in S} m_1 \cdot m_2(S) }Ê\\
&&=  \sum_{S:\theta \in S} \sum \{m_1(A)m_2(B):t_{\Theta_1 \vee \Theta_2}(A) \cap t_{\Theta_1 \vee \Theta_2}(B) =S\}.
\end{eqnarray*}
Now $\theta \in S$, where $S = t_{\Theta_1 \vee \Theta_2}(A) \cap t_{\Theta_1 \vee \Theta_2}(B)$, holds if and only if $\theta \in t_{\Theta_1 \vee \Theta_2}(A)$ and $\theta \in t_{\Theta_1 \vee \Theta_2}(B)$. Therefore we have
\begin{eqnarray*}
pl_{m_1 \cdot m_2}(\theta) = \sum_{A:\theta \in t_{\Theta_1 \vee \Theta_2}(A)} m_1(A) \cdot \sum_{B:\theta \in t_{\Theta_1 \vee \Theta_2}(B)} m_2(B).
\end{eqnarray*}
Further $\theta \in t_{\Theta_1 \vee \Theta_2}(A)$ holds if and only if $t_{\Theta_1}(\theta) \in A$ and, similarly, $\theta \in t_{\Theta_1 \vee \Theta_2}(B)$ holds if and only if $t_{\Theta_2}(\theta) \in B$. So we obtain finally
\begin{eqnarray*}
\lefteqn{ pl_{m_1 \cdot m_2}(\theta) = \sum_{A:t_{\Theta_1}(\theta) \in A} m_1(A) \cdot \sum_{B:t_{\Theta_2}(\theta) \in B} m_2(B) }\\
&&= pl_{m_1}(t_{\Theta_1}(\theta)) \cdot pl_{m_2}(t_{\Theta_2}(\theta)) = pl_{m_1} \cdot pl_{m_2}(\theta).
\end{eqnarray*}
This proves the claim.
\end{proof}

The potentials in $\Phi$ do not form an information algebra under the operations of combination and transport; in particular the Combination and Transport Axioms A4 and A5 do not hold in full generality. There are however weaker results about the interplay of combination and transport, which make the algebra $\Phi$ still an interesting and useful structure, especially for local computation (see Section \ref{subsec:LocComp}) and maximization (most probable elements, see Section \ref{sec:MostProb}). Some of these results follow here, for further ones, see Section \ref{sec:MostProb}

\begin{theorem} \label{th:WeakCombAxiom1}
Assume $\Theta_1 \bot \Theta_2 \vert \Lambda$ and $p_1$. $p_2$ and $p$ potentials in $\Phi$ with $d(p_1) = \Theta_1$, $d(p_2) = \Theta_2$ and $d(p) = \Lambda$. Then
\begin{eqnarray*}
\pi_\Lambda(p_1 \cdot p_2 \cdot p) = \pi_\Lambda(p_1) \cdot \pi_\Lambda(p_2) \cdot p.
\end{eqnarray*}
\end{theorem}

\begin{proof}
By definition, we have
\begin{eqnarray*}
\pi_\Lambda(p_1 \cdot p_2 \cdot p) = pl_{t_\lambda(m_{p_1 \cdot p_2 \cdot p})}.
\end{eqnarray*}
Further, for $\theta \in \Theta_1 \vee \Theta_2 \vee \Lambda$,
\begin{eqnarray*}
\lefteqn{m_{p_1 \cdot p_2 \cdot p}(\{\theta\}) = (p_1 \cdot p_2 \cdot p)(\theta) = p_1(t_{\Theta_1}(\theta)) \cdot p_2(t_{\Theta_2}(\theta)) \cdot p(t_\Lambda(\theta))} \\
&&= m_{p_1}(\{t_{\Theta_1}(\theta)\}) \cdot m_{p_2}(\{t_{\Theta_2}(\theta)\}) \cdot m_p(\{t_{\Lambda}(\theta)\}) = (m_{p_1} \cdot m_{p_2} \cdot m_p)(\theta),
\end{eqnarray*}
hence we conclude that $m_{p_1 \cdot p_2 \cdot p} = m_{p_1} \cdot m_{p_2} \cdot m_p$. Further, by the Combination Axiom A4 for set potentials, $\Theta_1 \bot \Theta_2 \vert \Lambda$ implies
\begin{eqnarray*}
t_\Lambda(m_{p_1} \cdot m_{p_2} \cdot m_p) = t_\Lambda(m_{p_1}) \cdot t_\Lambda(m_{p_2}) \cdot t_\Lambda(m_p).
\end{eqnarray*}
So, by Theorem \ref{th:CombOfPl} we obtain finally,
\begin{eqnarray*}
\lefteqn{\pi_\Lambda(p_1 \cdot p_2 \cdot p) = pl_{t_\Lambda(m_{p_1}) \cdot t_\Lambda(m_{p_2}) \cdot t_\Lambda(m_p)}} Ê\\
&&= pl_{t_\Lambda(m_{p_1})} \cdot pl_{t_\Lambda(m_{p_2})} \cdot pl_{t_\Lambda(m_p)} = \pi_\Lambda(p_1) \cdot \pi_\Lambda(p_2) \cdot p,
\end{eqnarray*}
since $pl_{m_p} = p$. This concludes the proof.
\end{proof}

\begin{theorem} \label{th:WeakCombAxiom2}
Assume $\Lambda \leq \Lambda_1 \leq \Theta$ and $p$ and $q$ potentials with $d(p) = \Theta$, $d(q) = \Lambda$. Then
\begin{eqnarray*}
\pi_{\Lambda_1}(p \cdot q) = \pi_{\Lambda_1}(p) \cdot q.
\end{eqnarray*}
\end{theorem}

\begin{proof}
Assume first that $\Lambda_1 = \Lambda$. Then, by definition, and since $m_{p \cdot q} = m_p \cdot m_q$ (see the proof of the previous theorem),
\begin{eqnarray*}
\pi_{\Lambda_1}(p \cdot q) = pl_{t_\Lambda(m_{p \cdot q})} = pl_{t_\Lambda(m_p \cdot m_q)},
\end{eqnarray*}
Now, $\Theta \bot \Lambda \vert \Lambda$. Therefore, by the Combination Axiom A4 for set potentials, be obtain $t_\Lambda(m_p \cdot m_q) = t_\Lambda(m_p) \cdot m_q$, hence (Theorem \ref{th:CombOfPl})
\begin{eqnarray} \label{eq:ProjOfComb}
\pi_{\Lambda_1}(p \cdot q) = pl_{t_\Lambda(m_p)} \cdot pl_{m_q} = \pi_\Lambda(p) \cdot q.
\end{eqnarray}

Now, if $\Lambda \leq \Lambda_1 \leq \Theta$, then 
\begin{eqnarray*}
\pi_{\Lambda_1}(p \cdot q) = \pi_{\Lambda_1}(p \cdot (\mathbf{1}_{\Lambda_1} \cdot q)).
\end{eqnarray*}
Now we apply (\ref{eq:ProjOfComb}), since $d(\mathbf{1}_{\Lambda_1} \cdot q) = \Lambda_1$, and obtain
\begin{eqnarray*}
\pi_{\Lambda_1}(p \cdot q) = \pi_{\Lambda_1}(p) \cdot \mathbf{1}_{\Lambda_1} \cdot q =  \pi_{\Lambda_1}(p) \cdot q,
\end{eqnarray*}
since the unit element $\mathbf{1}_{\Lambda_1}$ is absorbed by $\pi_{\Lambda_1}(p)$. This proves the claim.
\end{proof}

Note the special case $\Lambda_1 = \Lambda$, which is important.

\begin{theorem} \label{th:ProjAxiom}
Assume $\Lambda \leq \Theta$ and $p$ a potential with $d(p) \geq \Theta$. Then 
\begin{eqnarray*}
\pi_\Lambda(p) = \pi_\Lambda(\pi_\Theta(p)).
\end{eqnarray*}
\end{theorem}

\begin{proof}
Let $d(p) = \Theta_1$. Then, by definition $\pi_\Lambda(p) = pl_{t_\Lambda(m_p)}$. Further, since $\Theta_1 \bot \Theta \vert \Theta$, hence $\Theta_1 \bot \Lambda \vert \Theta$, by the Transport Axiom A5 for set potentials, we have $t_\Lambda(m_p) = t_\Lambda(t_\Theta(m_p))$. So, it follows that
\begin{eqnarray*}
\pi_\Lambda(P) = pl_{t_\Lambda(t_\Theta(m_p))} = \pi_\Lambda(pl_{t_\Theta(m_p)}) = \pi_\Lambda(\pi_\Theta(p))
\end{eqnarray*}
as claimed.
\end{proof}

We return now to normalization or scaling. For a set potential $m$ let $m^\downarrow$ be the the normalized potential, that is the bpa, associated with $m$, see (\ref{eq:ScaleSetPot}). Similarly, for a probability potential $p$ let $p^\downarrow$ be the associated normalized probability distribution, see (\ref{eq:ScaledProbPot}). We may define combination in the family $\Psi^\downarrow$ of bpa and the family of probability distributions $\Phi^\downarrow$ by
\begin{eqnarray*}
(m_1 \cdot m_2)^\downarrow, \quad (p_1 \cdot p_2)^\downarrow,
\end{eqnarray*}
where the dot denotes combination among set or probability potentials respectively. This is simply normalized combination like Dempster's rule. Similarly, we may define transport by
\begin{eqnarray*}
(t_\Theta(m))^\downarrow, \quad (\pi_\Theta(p))^\downarrow.
\end{eqnarray*}
We claim that for any set potentials $m_1,m_2$ or $m$ or probability potentials $p_1.p_2$ or $p$.
\begin{eqnarray*}
(m_1 \cdot m_2)^\downarrow = (m_1^\downarrow \cdot m_2^\downarrow)^\downarrow, \quad (p_1 \cdot p_2)^\downarrow = (p_1^\downarrow \cdot p_2^\downarrow)^\downarrow
\end{eqnarray*}
and 
\begin{eqnarray*}
(t_\Theta(m))^\downarrow = (t_\Theta(m^\downarrow))^\downarrow   \quad (\pi_\Theta(p))^\downarrow = (\pi_\Theta(p^\downarrow))^\downarrow.
\end{eqnarray*}
This has been proved in \cite{kohlas03} in the multivariate setting; it certainly holds also in the present case. These results say - loosely speaking - that normalization defines an information algebra homorphism, so that the normalized potentials form themselves information algebras.


\subsection{Commutative Families of Compatible frames} \label{subsec:CommFcF}

In this section commutative f.c.f $(\mathcal{F},\mathcal{R})$  are considered. Such a family of frames is characterized by the following two conditions:
\begin{enumerate}
\item $(\mathcal{F};\leq)$ is a lattice,
\item for all $\Theta,\Lambda \in \mathcal{F}$, we have $\Theta \bot \Lambda \vert \Theta \wedge \Lambda$.
\end{enumerate}
In particular the second condition is very strong, see Section \ref{subsec:CondInd} on this subject. Another way to express this condition is that if $\mathcal{P}_\Theta$ and $\mathcal{P}_\Lambda$ are the two partitions of $\Theta \vee \Lambda$ induced by the refinings of frames $\Theta$ and $\Lambda$ and if $v_\Theta$ and $v_\Lambda$ are the saturation mappings associated to these two partitions, defined by
\begin{eqnarray*}
v_\Theta(S) = \cup \{B \in \mathcal{P}_\Theta:B \cap S \not= \emptyset\}, v_\Lambda(S) = \cup \{B \in \mathcal{P}_\Lambda:B \cap S \not= \emptyset\},
\end{eqnarray*}
then $v_\Theta \circ v_\Lambda = v_\Lambda \circ v_\Theta$. Therefore it is said that the partitions commute and so therefore we call also such a f.c.f \textit{commutative}. Finally, commuting partitions are also called type I partitions \cite{graetzer78}. In \cite{menginwilson99} the condition above is called conditional independence (of frames). But in our development, conditional independence of frames means something more general, as explained in Section \ref{sec:FcF}.

Commutative f.c.f. are rather special. But the very important multivariate model belongs to this class. Or, in other words, commutative f.c.f are a generalization of multivariate models, keeping most of the desirable properties of it. In fact commutative f.c.f have interesting properties, not shared with general f.c.f. 

Reconsider the information algebra $\Psi$ of set potentials introduced in the previous section, but this time on a \textit{commutative} f.c.f $(\mathcal{F},\mathcal{R})$. In \cite{kohlas17} the following has been shown: If a new operator $\pi_\Lambda(m)$ for set potentials is defined by
\begin{eqnarray*}
\pi_\Lambda(m) = t_\Lambda(m) \textrm{ for}\ \Lambda \leq d(m)
\end{eqnarray*}
then we have a system where frames and set potential satisfy the following conditions:
\begin{description}
\item[B0] \textit{Lattice:} $(\mathcal{F};\leq)$ is a lattice.
\item[B1] \textit{Semigroup:} $(\Psi;\cdot)$ is a commutative semigroup.
\item[B2] \textit{Labeling:} $d(m_1 \cdot m_2) = d(m_1) \vee d(m_2)$ and $d(\pi_\Lambda(m)) = \Lambda$.
\item[B3] \textit{Unit and Null:} For all $\Theta \in \mathcal{F}$ there is a unit element $\mathbf{1}_\Theta$ with $d(\mathbf{1}_\Theta) = \Theta$ and a null element $\mathbf{0}_\Theta$ with $d(\mathbf{0}_\Theta) = \Theta$ such that
\begin{enumerate}
\item $m \cdot \mathbf{1}_\Theta  = m$ and $mÊ\cdot \mathbf{0}_\Theta = \mathbf{0}_\Theta$ if $d(m) = \Theta$,
\item if $\Lambda \leq \Theta = d(m)$, then $\pi_\Lambda(m) = \mathbf{0}_\Lambda$ if and only if $m = \mathbf{0}_\Theta$,
\item $\mathbf{1}_\Theta \cdot \mathbf{1}_\Lambda = \mathbf{1}_{\Theta \vee \Lambda}$.
\end{enumerate}
\item[B4] \textit{Projection:} If $\Lambda_1 \leq \Lambda_2 \leq d(m)$, then
\begin{eqnarray*}
\pi_{\Lambda_1}(\pi_{\Lambda_2}(m)) = \pi_{\Lambda_1}(m).
\end{eqnarray*}
\item[B5] \textit{Combination:} If $d(m_1) = \Theta$ and $d(m_2) = \Lambda$, then then
\begin{eqnarray*}
\pi_\Theta(m_1 \cdot m_2) = m_1 \cdot \pi_{\Theta \wedge \Lambda}(m_2).
\end{eqnarray*}
\end{description}
Such a system is called a \textit{valuation algebra}, since its axioms correspond to the older axiomatic systems introduced by \cite{shenoyshafer90}, see also \cite{kohlas03}.

According to the previous section, the map $m \rightarrow pl_m$ is a semi-group homomorphism from the set potentials $\Psi$ onto the probability potentials $\Phi$. Moreover, now we have for the transport operator of potentials also $\pi_\Lambda(pl_m) = pl_{t_\Lambda(m)} = pl_{\pi_\Lambda(m)}$, for $\Lambda \leq d(m)$, and where on the right hand side $\pi$ denotes the projection operator of bpa. Therefore, the map $m \mapsto pl_m$ is now a whole valuation algebra homomorphism.  This implies that in the case of commutative f.c.f, $\Phi$ is also a valuation algebra, satisfying conditions B0 to B5 above. The most popular and well-known version, indeed the almost uniquely considered instance in the literature, of this valuation algebra is the one relative to the multivariate model of frames; and this is also the version originally proposed in \cite{shenoyshafer90}. Usually this system is considered in the context of causal modeling or Bayesian networks rather than functional modeling like in this paper. Thus, this classical system of probability potentials can be extended to commutative f.c.f, see also \cite{kohlas17} where it is shown that the same effect is also valid for abstract information and valuation algebras. In the next section, we present, in a special multivariate setting, another valuation algebra related to probabilistic argumentation systems. To conclude, we remark that the units in the valuation algebra of set potentials have an additional property not shared by the valuation algebra of probability potentials, namely
\begin{eqnarray*}
\pi_\Lambda(\mathbf{1}_\Theta) = \mathbf{1}_\Lambda.
\end{eqnarray*}
This property is called stabilty. It allows to extend the projection operation backwards to the general transport operation, and thus reconstruct the original system of a generalized information algebra (axioms A1 to A6), see \cite{kohlas17}. Note that this is not possible for the valuation algebra of probability potentials, since there stability is not valid. With respect to normalization of set and probability potentials to bpa and probability distributions, the same as at the end of the previous section holds. In addition, the projection of a probability distribution is itself normalized, hence still a probability distribution. The same holds for bpas.


\subsection{Absolutely Continuous PAS} \label{subsec:AbsContPAS}

In this section we consider real-valued probabilistic argumentation systems, that is structures $(\Omega,\mathcal{A},P;X,\mathbb{R}^s)$ where $(\Omega,\mathcal{A},P)$ is a probability space and $X : \Omega \rightarrow \mathbb{R}^s$ a Borel-measurable function in s-dimensional real space, that is a random variable. Again, we consider $\omega \inÊ\Omega$ as an assumption, which, if valid, implies the value $X(\omega) \in \mathbb{R}^s$. We proceed in steps: First we examine the simple case of a one-dimensional value space $\mathbb{R}$. Next, we generalize to families of multidimensional real spaces. 

So, consider a tuple $(\Omega,\mathcal{A},P;X,\mathbb{R})$ as above. If $B$ is an element of the Borel-algebra $\mathcal{B}$ of $\mathbb{R}$, then the set $X^{-1}(B)$ belongs to $\mathcal{A}$. If we look at $(\Omega,\mathcal{A},P;X,\mathbb{R})$ as a probabilistic argumentation system in the sense of Section \ref{sec:PAS}, then $\mathbb{R}$ is considered as the set of possible values of some unknown magnitude and $X(\omega)$ the answer to this question, assuming $\omega$. Then we ask which set of assumptions $\omega \in \Omega$ support the hypothesisi that  the unknown magnitude is less than $x$, $\{\omega \in \Omega: X(\omega) \leq x\}$. Next we may also quantify the strength of this support as
\begin{eqnarray*}
sp(-\infty,x] = P\{\omega \in \Omega: X(\omega) \leq x\} = F(x).
\end{eqnarray*}
Of course $F(x)$ is simply the distribution function of the random variable $X$, interpreted however in our view as the degree of support induced by the PAS $(\Omega,\mathcal{A},P;X,\mathbb{R})$ on the hypothesis that the unknown magnitude 
is smaller than $x$.

For the following we require that $F(x)$ is \textit{absolutely continuous}, that is there exists a function $f(x)$ such that
\begin{eqnarray*}
F(x) = \int_\infty^x f(y) dy,
\end{eqnarray*}
where $f$ is measurable and the integral the Lebesgue integral or, if $f$ is continuous, the Rieman integral.

Things become now more interesting, if we assume that two (or more) structures $(\Omega_1,\mathcal{A}_1,P_1;X_1,\mathbb{R})$ and $(\Omega_2,\mathcal{A}_2,P_2;X_2,\mathbb{R})$ are available for the same unknown magnitude. How do we combine these two PAS into a new aggregated PAS? The approach is the same as in Section Ê\ref{subsec:indepPAS}. So, in each PAS an assumption must be valid, such that, combined, a pair $(\omega_1,\omega_2) \in \Omega_1 \times \Omega_2$ must be valid and consequently, the values are $X_1(\omega_1)$ and $X_2(\omega_2)$ conjointly. This however can only be the case if $X_1(\omega_1) = X_2(\omega_2)$. Therefore, the only consistent, non-contradictory pairs are those, which satisfy this condition, that is
\begin{eqnarray*}
\Omega = \{(\omega_1,\omega_2) \in \Omega_1 \times \Omega_2:X_1(\omega_1) = X_2(\omega_2)\}.
\end{eqnarray*}
If we assume the two PAS as \textit{independent}, then the pairs come from the product probability space $(\Omega_1 \times \Omega_2,\mathcal{A}_1 \times \mathcal{A}_2,P_1P_2)$. At this point the technical problem arises that the set $\Omega$ has probability zero, so that a conditioning of the product probability to the event $X_1(\omega_1) = X_2(\omega_2)$ in the usual way is not possible.

To circumvent this problem we resort to an approach used in \cite{kohlasmonney07}. Instead of considering random variables $X$ on probability spacees $(\Omega,\mathcal{A},P)$, we work directly with the distribution function of $X$, that is with the induced probability measure on $\mathbb{R}$. By our assumption that the distribution is absolutely  continuous, this probability is defined by the density function $f$ on $\mathbb{R}$. So, we consider PAS for an unknown magnitude in $\mathbb{R}$ essentially given by density functions $f$. If we denote the unknown magnitude by $x$, then, given two PAS of this form, we have the equations
\begin{eqnarray*}
x &=& \omega_1, \quad \omega_1 \sim f(w), \\
x &=& \omega_2, \quad \omega_2 \sim g(w).
\end{eqnarray*}
Here, consistency means $\omega_1 = \omega_2$ or $\omega_1 - \omega_2 = 0$, since the magnitude $x$ is unique. We use the variable-transformation
\begin{eqnarray*}
\xi_1 = \omega_1, \quad \xi_2 = \omega_2 - \omega_1,
\end{eqnarray*}
or, in the inverse transformation
\begin{eqnarray*}
\omega_1 = \xi_1, \quad \omega_2 = \xi_1 + \xi_2.
\end{eqnarray*}
The product density $f(w_1)g(w_2)$ for the pairs $(\omega_1,\omega_2)$ transforms then into
\begin{eqnarray*}
h(\xi_1,\xi_2) = f(\xi_1)g(\xi_1 + \xi_2).
\end{eqnarray*}
What we need now is the conditional density $h(\xi_1 \vert \xi_2 = 0)$ corresponding to the condition $\omega_1 - \omega_2 = 0$,
\begin{eqnarray*}
h(\xi_1 \vert \xi_2 = 0) = cf(\xi_1)g(\xi_1) = cf(x)g(x)
\end{eqnarray*}
since $x = \xi_1$. Here $c$ is a normalization constant. This shows that combining absolutely continuous PAS results essentially in multiplying the densities. This holds also in a more general setting to be presented next. And this leads to a valuation algebra of densities, generalizing the valuation algebra of (discrete) probability potentials as shown below.

Consider a finite or countable set $r$ of real-valued variables $x_i$, $i \in r$ and suppose the common value of these variables $x_i$ are the unknowns, we want to determine. Let $s,t,\ldots$ denote finite subsets of $r$ and $x_s : s \rightarrow \mathbb{R}$ denote $s$-tuples of real values. Then $\mathbb{R}^s$ is the corresponding linear $s$-space of these $s$-tuples. We assume now that information about the $x_i$ is given by some $s$-densities on the space $\mathbb{R}^s$ for some subsets $s$ of $r$. More precisely, lef $f$ and $g$ be two density functions, the first one on $\mathbb{R}^s$, the second one on $\mathbb{R}^t$. We want to aggregate these two PAS into an aggregated one, generalizing the technique used above to combine two PAS relative to $\mathbb{R}$. Let $\omega_1$ be a $\mathbb{R}^s$ random variable with density function $f$ and $\omega_2$ be a $\mathbb{R}^t$ random variable with density function $g$. Consider two PAS, one relating to $s$, the other one to $t$ given by
\begin{eqnarray*}
x_s &=& \omega_1, \quad \omega_1 \sim f, \\
x_t &=& \omega_2, \quad \omega_2 \sim g, 
\end{eqnarray*}
We may decompose these equations into
\begin{eqnarray*}
x_{s-t} &=& \omega_{1,s-t}, \\
x_{s \cap t} &=& \omega_{1,s \cap t}, \\
x_{s \cap t} &=& \omega_{2,s \cap t}, \\
x_{t-s} &=& \omega_{2,t-s}, 
\end{eqnarray*}
The consistency condition between $\omega_1$ and $\omega_2$ is now $\omega_{1,s \cap t} = \omega_{2,s \cap t}$ or $\omega_{2,s \cap t} - \omega_{1,s \cap t} = 0$. Following the model above for $\mathbb{R}$ we use the transformation
\begin{eqnarray*}
\xi_1 = \omega_1, \quad \xi_{2,s \cap t} = \omega_{2,s \cap t} - \omega_{1,s \cap t}, \quad \xi_{2,t-s} = \omega_{2,t-s},
\end{eqnarray*}
or, inversely,
\begin{eqnarray*}
\omega_1 = \xi_1, \quad \omega_{2,s \cap t} = \xi_{1,s \cap t} + \xi_{2,s \cap t}, \quad \omega_{2,t-s} = \xi_{2,t-s},
\end{eqnarray*}
Again, assuming stochastic independence between the two PAS, the pairs $(\omega_1,\omega_2)$ have the product density $f(\omega_1)g(\omega_2)$. Then the pairs $(\xi_1,\xi_2)$ have the density
\begin{eqnarray*}
h(\xi_1,\xi_2) = f(\xi_1)g(\xi_{1,s \cap t} + \xi_{2,s \cap t},\xi_{2,t-s}).
\end{eqnarray*}
The conditional density of $\xi_1$ given $\xi_{2,s \cap t} = 0$, corresponding to the consistency condition $\omega_{2,s \cap t} - \omega_{1,s \cap t} = 0$ is then
\begin{eqnarray*}
h(\xi_1 \vert \xi_{2,s \cap t = 0}) = cf(\xi_1)g(\xi_{1,s \cap t},\xi_{2,t-s}).
\end{eqnarray*}
Then, since $x_s = \omega_1 = \xi_1$ and $x_{t-s} = \omega_{2,t-s} = \xi_{2,t-s}$ we obtain for the density of the support of the unknown magnitude $x_{s \cup t}$
\begin{eqnarray*}
h(x_{s \cup t}) = cf(x_s)g(x_t).
\end{eqnarray*}
Again, we find essentially for the aggregation of the two PAS the multiplication law of the two densities of the two PAS.

Projecting a PAS given by a density on $\mathbb{R}^s$ to $\mathbb{R}^t$ for $t \subseteq s$ gives a new PAS with density
\begin{eqnarray} \label{eq:ProjOfDens}
(\pi_t(f))(x_t) = \int_{-\infty}^{+ \infty} f(x_t,x_{s-t})dx_{s-t}.
\end{eqnarray}
These operations of combination and projection of absolutely continuous PAS on real spaces $\mathbb{R}^s$ give rise to a valuation algebra of density functions, similar to the one of probability potentials. 

In fact, let $\Psi$ denote the family of non-negative real-valued continuous functions $f : \mathbb{R}^s \rightarrow \mathbb{R}^+ \cup \{0\}$ for $s \subseteq r$ with \textit{finite} integral
\begin{eqnarray*}
\int_{- \infty}^{+ \infty} f(x_s)dx_s < \infty.
\end{eqnarray*}
Define the following operations in $\Psi$: 
\begin{enumerate}
\item \textit{Labeling:} $d(f) = s$ if $f$ is defined on $\mathbb{R}^s$,
\item \textit{Combination:} If $d(f) = $s and $d(g) = t$, then $f \cdot g$ is defined by $(f \cdot g)(x) = f(x_s)g(x_t)$ for $x \in \mathbb{R}^{s \cup t}$.
\item \textit{Projection:} if $d(f) = s$ and $t \subseteq s$, then $\pi_t(f)$ is defined by (\ref{eq:ProjOfDens}).
\end{enumerate}
Note that all operations are well defined in the sense that both combination and projection result in elements of $\Psi$. Let  $\mathcal{F} = \{\mathbb{R}^s:s \subseteq r\}$. With these operations $\Psi$ forms a valuation algebra, that is satisfies the following axioms:
\begin{description}
\item[C0] \textit{Lattice:} $(\mathcal{F};\leq)$ is a lattice with $\mathbb{R}^t \leq \mathbb{R}^s$ if $t \subseteq s$..
\item[C1] \textit{Semigroup:} $(\Psi;\cdot)$ is a commutative semigroup.
\item[C2] \textit{Labeling:} $d(f_1 \cdot f_2) = d(f_1) \vee d(f_2)$ and $d(\pi_t(f)) = \mathbb{R}^t$.
\item[C3] \textit{Null:} For all $\mathbb{R}^s$ there is a null element $\mathbf{0}_s$ with $d(\mathbf{0}_s) = \mathbb{R}^s$  such that
\begin{enumerate}
\item $mÊ\cdot \mathbf{0}_s = \mathbf{0}_s$ if $d(f) = \mathbb{R}^s$,
\item if $\mathbb{R}^t \leq \mathbb{R}^s = d(f)$, then $\pi_t(f) = \mathbf{0}_t$ if and only if $f = \mathbf{0}_s$.
\end{enumerate}
\item[B4] \textit{Projection:} If $t_1 \leq t_2 \leq d(f)$, then
\begin{eqnarray*}
\pi_{t_1}(\pi_{t_2}(f)) = \pi_{t_1}(f).
\end{eqnarray*}
\item[B5] \textit{Combination:} If $d(f_1) = s$ and $d(f_2) =t$, then then
\begin{eqnarray*}
\pi_s(f_1 \cdot f_2) = f_1 \cdot \pi_{s \cap t}(f_2).
\end{eqnarray*}
\end{description}
Note that $(\mathcal{F};\leq)$ is even a \textit{dstributive} lattice. The null element is defined by $\mathbf{0}_s(x) = 0$ for all $x \in \mathbb{R}^s$. Instead of continuous functions, we might also consider Lebesgue-measurable functions. An interesting subalgebra of this valuation algebra is the algebra of Gaussian densities \cite{kohlas03}, see also \cite{kohlasmonney07,poulykohlas11}. Note that the function $\mathbf{1}_s(x) = 1$ for all $x \in \mathbb{R}^s$ is a unit element for combination, but is not integrable and therefore does not belong to $\Psi$. However, $\Psi$ can be extended to incorporate these unit elements (and other elements), see \cite{kohlas03}. More on this kind of probabilistic argumentation structures for statistical analysis can be found in \cite{km04}.



\section{Conditioning} \label{sec:Cond}


\subsection{The Nature of Conditioning} \label{subsec:NatCond}

In classical (discrete) probability theory, conditioning refers to changing a probability of an event if another event occurs, leading to a conditional probability. Slightly more generally, if a (discrete) probability distribution is given, observing an event leads in this way to a conditional probability distribution (given the observed event). If the probability distribution is considered as induced by a (precise) PAS, then the process of conditioning can be seen as a combination of information, namely of the precise PAS with the deterministic PAS determined by the event. This view puts conditioning into a wider perspective: Conditioning is simply the combination of a (precise) PAS with any other PAS. It is thus a process carried out within the injformation algebra of bpas. This point of view will be developed in this section, whereas in the following section, conditioning will be limited to a more restricted view.

Conditioning can be captured algebraically by an operation of combination between a probability and a set potential, written as $p \cdot m$. Consider the information algebra of bpas over some f.c.f $(\mathcal{F},\mathcal{R})$ (see Section \ref{subsec:AlgOfProbPot}). Within this framework consider the set potential $m_p$ associated with the probability potential $p$ and any other set potential $m$, both on the same domain $\Theta$. Then we have
\begin{eqnarray}
(m_p \cdot m)(\{\theta\}) = \sum_{S:\theta \in S} p(\theta)m(S).
\end{eqnarray}
Recall that
\begin{eqnarray*}
pl_m(\theta) = \sum_{S:\theta \in S} m(S).
\end{eqnarray*}
Thus we conclude that
\begin{eqnarray*}
(m_p \cdot m)(\{\theta\}) = p(\theta)pl_m(\theta),
\end{eqnarray*}
Suppose next that $m$ is a bpa on a frame $\Lambda \leq \Theta$. Then $m_p \cdot m$ is a bpa on the frame $\Theta$ and this combined bpa is still a set potential non-null only on singleton sets, hence essentially a probability potential. In view of these results we define the combination $p \cdot m$ by
\begin{eqnarray*}
p \cdot m(\theta) = (m_p \cdot m)(\{\theta\})
\end{eqnarray*}
for all $\theta \in d(p9$. We may also see this as a map $\Phi \times \Psi \rightarrow \Phi$, defining combination between a probability potential $p$ and a set potential, $m$, provided that $d(p) \geq d(m)$. Here $\Psi$ denotes as before the family of set potentials and $\Phi$ the family of potentials over an f.c.f $(\mathcal{F},\mathcal{R})$.

Note that $m_p \cdot m$ with $d(m) = \Lambda$ arbitrary is, in general, no more a probability potential. However, since
\begin{eqnarray*}
t_\Theta(m_p \cdot m) = m_p \cdot t_\Theta(m),
\end{eqnarray*}
it follows that $t_\Theta(m_p \cdot m)$ is essentially a probability potential for any $m$. This last case covers also the first two cases. In fact, $t_\Theta(m)$ is a set potential on domain $\Theta$, thus, as above,
\begin{eqnarray}
t_\Theta(m_p \cdot m)(\{\theta\}) = p(\theta)pl_{t_\Theta(m)}(\theta).
\end{eqnarray}
Note that in any of these cases this represents finally a combination operation between two probability potentials (where at least one of them is not normalized). All this may be called \textit{conditioning} of a probability potential $p$ on another information represented by some bpa $m$.

If in particular $m$ is a \textit{deterministic} bpa, that is $m(B) = 1$ for some $B \subseteq \Lambda$ and $m(A) = 0$ for all other subsets of $\Lambda$, then $pl_{t_\Theta(m)}(\theta) = 1$ for all $\theta \in t_\Theta(B)$. If we normalize $t_\Theta(p \cdot m)$ in this case, we obtain
\begin{eqnarray*}
p'(\theta) = \left\{ \begin{array}{l} \frac{p(\theta)}{\sum_{\theta \in t_\Theta(B)} p(\theta)}, \textrm{ if}\ \theta \in t_\Theta(B) \\ 0, \textrm{ otherwise}. \end{array} \right.
\end{eqnarray*}
So, $p'$ is the ordinary, classical conditional probability distribution of $p$, given the event $t_\Theta(B)$! In particular, if $d(m) = \Theta$, then $p'$ is the conditional probability distribution of $p$ given $B$. 

An important special case is the following: Let $p$ be a probability potential on a frame $\Theta$ and $\Lambda \leq \Theta$. Then consider the probability potential
\begin{eqnarray} \label{eq:DefOfConditional_2}
p'(\theta) =  \frac{p(\theta)}{\pi_\Lambda(p)(t_\Lambda(\theta))}.
\end{eqnarray}
For any $\lambda \in \Lambda$ and $\theta \in t_\Theta(\lambda)$, we have then
\begin{eqnarray*}
p'(\theta) =  \frac{p(\theta)}{p(t_\Theta(\lambda))},
\end{eqnarray*}
that is, $p'(\theta)$ is the conditional probability of $\theta \in t_\Theta(\lambda)$, given $t_\Theta(\theta) = \lambda$. We call the probability potential $p'$ as defined in (\ref{eq:DefOfConditional_2}) the \textit{conditional} of $p$ in $\Theta$, given $\Lambda$ and write $p' = p_{\Theta \vert \Lambda}$. Note that it is a potential on $\Theta$, that is $d(p_{\Theta \vert \Lambda}) = \Theta$. 

This reduces to a more familiar notion if we consider the special case of probability distributions on multivariate models. Let $r,s,t$ be index sets
 sucht that $r = s \cup t$, $s \cap t = \emptyset$ and let $\Theta_r$, $\Theta_s$ and $\Theta_t$ be corresponding frames,
 \begin{eqnarray*}
\Theta_r = \prod_{i \in r} \Theta_i, \quad \Theta_r = \prod_{i \in s} \Theta_i, \quad \Theta_r = \prod_{i \in t} \Theta_i,
\end{eqnarray*}
where $\Theta_i$ are finite sets, the frames of variables $X_i$. Consider now a probability distribution $p$ over frame $\Theta_r$. For $\theta_r \in \Theta_r$, $\theta_s \in \Theta_s$ and $\theta_t \in \Theta_t$ with $\theta_r = (\theta_s,\theta_t)$, define
\begin{eqnarray} \label{eq:CondPropDistr}
p'(\theta_r)  = p'(\theta_s,\theta_t) = \frac{p(\theta_s,\theta_t)}{\sum_{\theta_t \in \Theta_t} p(\theta_s,\theta_t)}-
\end{eqnarray}
Clearly, the denominator is $t_{\Theta_s}(p)(\theta_s)$ and $\theta_s = t_{\Theta_s}(\theta_s,\theta_t)$. So, we see that $p'$ here corresponds to (\ref{eq:DefOfConditional_2}). But $p'$ defined according to (\ref{eq:CondPropDistr}) is nothing else than the usual multivariate family of conditional probability distributions, often written as $p'(\theta_s,\theta_t) = p(\theta_s \vert \theta_t)$. Therefore, the probability potentials defined by (\ref{eq:DefOfConditional_2}) are generalizations of the classical concept of multivariate conditional probability distributions related to multivariate models to the more general concept of probability distributions on a $f.c.f$. This concept is of some interest and will be studied in the next section.


\subsection{Conditionals and Continuation} \label{subsec:CondAndCont}

Consider a f.c.f $(\mathcal{F},\mathcal{R})$, where $(\mathcal{F};\leq)$ is a join-semilattice and consider the family of probability potentials $\Phi$ on it. Recall that this system of potentials is closed under combination and projection, that is transport to $\Lambda \leq d(p)$. When we use here the notation $\pi_\Lambda(p)$ it is always implicitly assumed that $\Lambda \leq d(p)$. 

As a preparation to the study of conditionals as defined in the previous section, we define the support set $supp(p) = \{\theta \in d(p):p(\theta > 0\}$ of a probability potential $p$. Associated with any potential $p$ we define also the potential
\begin{eqnarray*}
f_p(\theta) = \left\{ \begin{array} {l} 1, \textrm{ if}\ \theta \in supp(p) \\ 0, \textrm{ otherwise}. \end{array} \right.
\end{eqnarray*}
These potentials $f_p$ are in fact the indicator functions of the sets $supp(p)$, they have the same domain as $p$, $d(f_p) = d(p)$ and are \textit{idempotent}, 
\begin{eqnarray*}
f_p \cdot f_p = f_p.
\end{eqnarray*}
Furthermore, if $p$ and $q$ are two probability potentials, then
\begin{eqnarray*}
f_p \cdot f_q = f_{p \cdot q}.
\end{eqnarray*}
Between these idempotent elements, a partial order (p.o) can be defined by
\begin{eqnarray*}
f_p \leq f_q \textrm{ if}\ f_p \cdot f_q = f_q.
\end{eqnarray*}
This p.o is in fact a join-semilattice, that is,
\begin{eqnarray*}
f_p \cdot f_q = \sup\{f_p,f_q\},
\end{eqnarray*}
since $f_p,f_q \leq g$ for some idempotent $g$ implies $f_p \cdot f_q \leq g \cdot g = g$. We shall write $\sup\{f_p,f_q\}$ as $f_p \vee f_q$. This p.o between the idempotent elements $f_p$ represents in fact also a p.o between support sets $supp(p)$, since we may define $supp(p) \leq supp(q)$ iff $f_p \leq f_q$. If $p$ and $q$ have the same domain, then obviously $supp(p) \leq supp(q)$ if and only if $supp(p) \supseteq supp(q)$. In passing, we remark that the system of the subsets of frames of an f.c.f, or equivalently, of idempotents form themselves an \textit{idempotent} generalized information algebra, which is isomorphic to the subalgebra of deterministic set potentials, also called a set algebra, see \cite{kohlas17}.

Write $p \equiv q$ if $supp(p) = supp(q)$. This is an equivalence relation and moreover a congruence relative to the operations of combination and projection, that is $p \equiv q$ implies 
\begin{eqnarray*}
d(p) &=& d(q), \\
p \cdot u &\equiv& q \cdot u, \textrm{ for any potential}\ u,\\
\pi_\Lambda(p) &\equiv& \pi_\Lambda(q), \textrm{ for any frame}\ \Lambda \leq d(p) = d(q).
\end{eqnarray*}
Let $[p]$ denote the equivalence classes of this equivalence relation. Each such class turns out to be a commutative group, the unit of $[p]$ is $f_p$ and the inverse of $p$ is defined by
\begin{eqnarray*}
p^{-1}(\theta) = \left\{ \begin{array}{l} 1/p(\theta), \textrm{ if}\ \theta \in supp(p), \\ 0, \textrm{Êotherwise}. \end{array} \right.
\end{eqnarray*}
Thereby, in the commutative semigroup of probability potentials over a f.c.f, a partial division is defined. This is an instance of a general theory of division in \textit{regular} semigroups, \cite{cliffordpreston67,croisot53}. It is also an extension of the theory of regular valuation algebras \cite{kohlas03,kohlas17} from the multivariate setting to potentials over a f.c.f. The equivalence classes form an idempotent and commutative semigroup if we define
\begin{eqnarray*}
[p]Ê\cdot [q] = [p \cdot q]
\end{eqnarray*}
and the p.o of the idempotents $f_p$ carries also over to these classes: $[p] \leq [q]$ iff $[p] \cdot [q] = [p \cdot q] = [p] \vee [q]\}$.  Note that $[p^{-1}] = [p]$.

The following Lemma is very important for the the subsequent study of conditionals. It is in fact a statement about support sets of potentials.

\begin{lemma} \label{le:ProjOrder}     
\begin{enumerate}
\item Assume $d(p) = d(q)$. Then $[p] \leq |q]$ implies $[\pi_\Lambda(p)] \leq [\pi_\Lambda(q)]$,
\item For all $\Lambda \leq d(p)$ we have $[\pi_\Lambda(p)] \leq [p]$.
\end{enumerate}
\end{lemma}

\begin{proof}
1.) $[p] \leq [q]$ implies $supp(q) \subseteq supp(p)$. Consider an element $\lambda \in supp(\pi_\Lambda(q))$. Then, if $d(q) = \Theta$,
\begin{eqnarray*}
\pi_\Lambda(q)(\lambda) = \sum_{\theta \in t_\Theta(\lambda)} q(\theta) > 0.
\end{eqnarray*}
Thus, there is at least one element $\theta' \in t_\Theta(\lambda)$ such that $q(\theta') > 0$. Then, since $supp(q) \subseteq supp(p)$, we have also $p(\theta') > 0$, hence $\pi_\Lambda(p)(\lambda) > 0$. This means that $supp(\pi_\Lambda(q)) \subseteq supp(\pi_\Lambda(p))$, hence $[\pi_\Lambda(p)] \leq \pi_\Lambda(q)]$.

2.) We show that $f_{\pi_\Lambda(p)} \leq f_p$, that is $f_{\pi_\Lambda(p)}(t_\Lambda(\theta))f(\theta) = f(\theta)$. This holds exactly if $f(\theta) = 1$ implies $f_{\pi_\Lambda(p)}(t_\Lambda(\theta)) = 1$. Now, $f_p(\theta) = 1$ means that $\theta \in supp(p)$, such that $\pi_\Lambda(p)(\lambda) > 0$ if $\lambda = t_\Lambda(\theta)$. But then $f_{\pi_\Lambda(p)}(\lambda) = f_{\pi_\Lambda(p)}(t_\Lambda(\theta)) = 1$. 
\end{proof}

Note that if $[p] \leq [q]$, then $f_p \cdot f_q = f_q$, hence $f_p \cdot q = f_p \cdot f_q \cdot q = f_q \cdot q = q$. So, we conclude that $[p] \leq [q]$ implies $f_{\pi_\Lambda(p)} \cdot \pi_\Lambda(q) = \pi_\Lambda(q)$ and $f_{\pi_\Lambda(p)} \cdot p = p$. These observations will be useful later.

For a probability potential $p$ and $\Lambda \leq \Theta \leq d(p)$ we define the potential
\begin{eqnarray} \label{eq:DefOfConditional}
p_{\Theta \vert \Lambda} = \pi_\Theta(p) \cdot (\pi_\Lambda(p))^{-1}.
\end{eqnarray}
If $d(p) = \Theta$, then this corresponds to (\ref{eq:DefOfConditional_2}). Therefore, we call $p_{\Theta \vert \Lambda}$ the \textit{conditional} of $p$ for $\Theta$ given $\Lambda$. In considering a conditional $p_{\Theta \vert \Lambda}$ we always implicitly assume that $\Lambda \leq \Theta \leq d(p)$. Such conditionals have bee studied in detail in \cite{kohlas03} in the case of a valuation algebra in a multivariate setting. We shall show now, that the results obtained there, which generalize well-know results of classical multivariate conditional probability distributions, extend also to the present case of potentials over a f.c.f.

The following lemma gives some basic properties of conditionals.

\begin{lemma} \label{le:BasPropOfCiond}
The following statements are valid:
\begin{enumerate}
\item $[\pi_\Lambda(p)] \leq [p_{\Theta \vert \Lambda}]$,
\item $\pi_\Lambda(p_{\Theta \vert \Lambda}) = f_{\pi_\Lambda(p)}$,
\item If $\Lambda_1 \leq \Lambda_2 \leq \Theta$ then $p_{\Theta \vert \Lambda_2} = p_{\Theta \vert \Lambda_1} \cdot p_{\Lambda_1 \vert \Lambda_2}$.
\item If $\Lambda \leq \Theta_1 \leq \Theta$ then $\pi_{\Theta_1}(p_{\Theta \vert \Lambda}) = p_{\Theta_1 \vert \Lambda}$,
\item If $\Theta_1 \leq \Lambda_1,\Lambda_2 \leq \Theta$ then $\pi_{\Lambda_1}(p_{\Theta \vert \Lambda_2} \cdot p_{\Lambda_2 \vert \Theta_1}) = p_{\Lambda_1 \vert \Theta_1}$,
\item If $d(p_2) = \Lambda$, then $(\pi_\Theta(p_1) \cdot p_2)_{\Theta \vert \Lambda} = p_{1 \Theta \vert \Lambda} \cdot f_{p_2}$.
\end{enumerate}
\end{lemma}

\begin{proof}
1.) By definition and since the p.o of equivalence classes $[p]$ is a join-semilattice with $[p \cdot q] = [p] \vee [q]$,  we have
\begin{eqnarray*}
[p_{\Theta \vert \Lambda}] = [\pi_\Theta(p) \cdot (\pi_\Lambda(p))^{-1}] = [\pi_\Theta(p)] \vee [(\pi_\Lambda(p))^{-1}] \geq [(\pi_\Lambda(p))^{-1}] = [\pi_\Lambda(p)].
\end{eqnarray*}

2.) Again, by definition,
\begin{eqnarray*}
\pi_\Lambda(p_{\Theta \vert \Lambda}) = \pi_\Lambda(\pi_\Theta(p) \cdot (\pi_\Lambda(p))^{-1}) = \pi_\Lambda(\pi_\Theta(p)) \cdot (\pi_\Lambda(p))^{-1}
\end{eqnarray*}
by Theorem \ref{th:WeakCombAxiom1}, since $\Theta \bot \Lambda \vert \Lambda$. But $\Lambda \leq \Theta$ implies $\pi_\Lambda(\pi_\Theta(p)) = \pi_\Lambda(p)$( see (Theorem \ref{th:ProjAxiom}), hence 
\begin{eqnarray*}
\pi_\Lambda(p_{\Theta \vert \Lambda}) = \pi_\Lambda(p) \cdot (\pi_\Lambda(p))^{-1} = f_{\pi_\Lambda(p)}.
\end{eqnarray*}

3.) We have $p_{\Theta \vert \Lambda_2} = \pi_\Theta(p) \cdot (\pi_{\Lambda_2}(p))^{-1}$. Since $\pi_{\Lambda_1}(p) \cdot (\pi_{\Lambda_1}(p))^{-1} = f_{\pi_{\Lambda_1}(p)} = f_{\pi_{\Lambda_1}(\pi_\Theta(p))}$ (see Theorem \ref{th:ProjAxiom}), we have by Lemma \ref{le:ProjOrder}
\begin{eqnarray*}
p_{\Theta \vert \Lambda_2} = (\pi_\Theta(p) \cdot (\pi_{\Lambda_1}(p))^{-1}) \cdot (\pi_{\Lambda_1}(p) \cdot (\pi_{\Lambda_2}(p))^{-1} )
= p_{\Theta \vert \Lambda_1} \cdot p_{\Lambda_1 \vert \Lambda_2}.
\end{eqnarray*}

4.)  Here we start with $\pi_{\Theta_1}(p_{\Theta \vert \Lambda}) = \pi_{\Theta_1}(\pi_\Theta(p) \cdot (\pi_\Lambda(p))^{-1})$. Since $\Lambda \leq \Theta_1 \leq \Theta$, we have $\Theta_1 \bot \Theta \vert \Theta_1$, hence $\Lambda \bot \Theta \vert \Theta_1$. It follows from Theorem \ref{th:WeakCombAxiom1} and Theorem \ref{th:ProjAxiom} that
\begin{eqnarray*}
\pi_{\Theta_1}(p_{\Theta \vert \Lambda}) = \pi_{\Theta_1}(p) \cdot (\pi_\Lambda(p))^{-1} = p_{\Theta_1 \vert \Lambda}.
\end{eqnarray*}

5.)  We have $p_{\Theta \vert \Lambda_2} \cdot p_{\Lambda_2 \vert \Theta_1} = p_{\Theta \vert \Theta_1}$ by item 3 above and $\pi_{\Lambda_1}(p_{\Theta \vert \Theta_1}) = p_{\Lambda_1 \vert \Theta_1}$ by item 4.

6.) The definition of conditionals gives
\begin{eqnarray*}
(\pi_\Theta(p_1) \cdot p_2)_{\Theta \vert \Lambda} = \pi_\Theta(\pi_\Theta(p_1) \cdot p_2) \cdot (\pi_\Lambda(\pi_\Theta(p_1) \cdot p_2))^{-1}.
\end{eqnarray*}
Applying Theorems \ref{th:WeakCombAxiom1} and \ref{th:ProjAxiom} we obtain
\begin{eqnarray*}
\lefteqn{(\pi_\Theta(p_1) \cdot p_2)_{\Theta \vert \Lambda} = (\pi_\Theta(p_1) \cdot p_2) \cdot (\pi_\Lambda(p_1) \cdot p_2)^{-1}}Ê\\
&&= (\pi_\Theta(p_1) \cdot (\pi_\Lambda(p_1))^{-1}) \cdot (p_2 \cdot p_2^{-1})  = p_{1 \Theta \vert \Lambda} \cdot f_{p_2}.
\end{eqnarray*}
This concludes the proof
\end{proof}

If we multiply both sides of the definition  (\ref{eq:DefOfConditional}) by $\pi_\Lambda(p)$, then we obtain $\pi_\Theta(p) \cdot f_{\pi_\Lambda(p)} = p_{\Theta \vert \Lambda} \cdot \pi_\Lambda(p)$. Due to Lemma \ref{le:ProjOrder} we conclude then that
\begin{eqnarray} \label{eq:Continuation}
\pi_\Theta(p) = p_{\Theta \vert \Lambda} \cdot \pi_\Lambda(p).
\end{eqnarray}
In the words of \cite{shafer96} the conditional $p_{\Theta \vert \Lambda}$ \textit{continues} $\pi_\Lambda(p)$ from $\Lambda$ to $\Theta$ and we call (\ref{eq:Continuation}) the \textit{continuation} property of the conditional.


\subsection{Factorization of Potentials} \label{subsec:CondIndPot}

Based on factorizations of a probability potential $p$, a new relation between frames in a f.c.f is introduced and studied in this section. We extend for this purpose the definition of a conditional slightly. So, if $p$ is a probability potental with $d(p) \geq \Theta \vee \Lambda$, then we define
\begin{eqnarray*}
p_{\Theta \vert \Lambda} = \pi_{\Theta \vee \Lambda}(p) \cdot (\pi_\Lambda(p))^{-1}.
\end{eqnarray*}
Hence, we do no more assume the $\Lambda \leq \Theta$. But note that $p_{\Theta \vert \Lambda} = p_{\Theta \vee \Lambda \vert \Lambda}$ so that the new definition is an extension of the old one.

\begin{definition} \label{def:CondIndepPot}
Let $\Theta_1 \bot \Theta_2 \vert \Lambda$. If $q_1$ and $q_2$ are two probability potentials such that $d(q_1) = \Theta_1 \vee \Lambda$ and $d(q_2) = \Theta_2 \vee \Lambda$. Then we call the probability potential $q_1$ and $q_2$ conditionally independent given $\Lambda$, and write $q_1 \bot q_2 \vert \Lambda$.
\end{definition}

The following theorem gives an interpretation of the meaning of this concept.

\begin{theorem} \label{th:CondIndOfPot}
$q_1 \bot q_2 \vert \Lambda$ implies
\begin{enumerate}
\item $\pi_{\Theta_1 \vee \Lambda}(q_1 \cdot q_2) = q_1 \cdot \pi_\Lambda(q_2)$ and $\pi_{\Theta_2 \vee \Lambda}(p) = q_2 \cdot \pi_\Lambda(q_1)$,
\item $\pi_\Lambda(q_1 \cdot q_2) = \pi_\Lambda(q_1) \cdot \pi(_\Lambda(q_2)$.
\end{enumerate}
\end{theorem}

\begin{proof}
1.) We have by definition
\begin{eqnarray*}
\pi_{\Theta_1 \vee \Lambda}(q_1 \cdot q_2) = pl_{t_{\Theta_1 \vee \Lambda}(m_{q_1} \cdot m_{q_2})}.
\end{eqnarray*}
Then $\Theta_1 \vee \Lambda \bot \Theta_2 \vee \Lambda \vert \Theta_1 \vee \Lambda$ implies $t_{\Theta_1 \vee \Lambda}(m_{q_1} \cdot m_{q_2}) = m_{q_1} \cdot t_{\Theta_1 \vee \Lambda}(m_{q_2})$. Further, from $\Theta_1 \bot \Theta_2 \vert \Lambda$ follows $\Theta_1 \vee \Lambda \bot \Theta_2 \vee \Lambda \vert \Lambda$ and therefore by the Transport Axiom A4 for set potentials 
\begin{eqnarray*}
t_{\Theta_1 \vee \Lambda}(m_{q_2}) = t_{\Theta_1 \vee \Lambda}(t_\Lambda(m_{q_2})) = \mathbf{1}_{\Theta_1 \vee \Lambda} \cdot t_\Lambda(m_{q_2}).
\end{eqnarray*}
This gives us then $t_{\Theta_1 \vee \Lambda}(m_{q_1} \cdot m_{q_2}) = m_{q_1} \cdot t_\Lambda(m_{q_2})$, since the unit $\mathbf{1}_{\Theta_1 \vee \Lambda}$ is absorbed by the first factor. So, we obtain finally that
\begin{eqnarray*}
\pi_{\Theta_1 \vee \Lambda}(q_1 \cdot q_2) = pl_{m_{q_1} \cdot t_\Lambda(m_{q_2})} = q_1 \cdot \pi_\Lambda)q_2).
\end{eqnarray*}
This proves the first part.

2.) follows from then Theorem \ref{th:WeakCombAxiom1}.
\end{proof}

These results show that in case of conditional independence of potentials $q_1$ and $q_2$ given $\Lambda$, the part of information in $p = q_1 \cdot q_2$ relating to frame $\Theta_1 \vee \Lambda$ depends only on the information $q_2$ relating to frame $\Lambda$ and the part of information $p$ relating to frame $\Lambda$ depends only on the information in $q_1$ and $q_2$ relating to this same frame. As we shall see later, this has important computational consequences, see Section \ref{sec:LocComp}.

The following theorem shows that conditionals are closely related to conditional independence. The results of this theorem are a generalization of results for conditionals of probability potentials, or more generally, valuations in a regular valuation algebra in a \textit{multivariate} framework \cite{kohlas03}.

\begin{theorem}Ê\label{th:EquivOfCondIndep}
Assume $\Theta_1 \bot \Theta_2 \vert \Lambda$. Then the following statements are all equivalent:
\begin{enumerate}
\item $p = q_1 \cdot q_2$, where $d(q_1) = \Theta_1 \vee \Lambda$ and $d(q_2) = \Theta_2 \vee \Lambda$,
\item $p = p_{\Theta_1 \vert \Lambda} \cdot p_{\Theta_2 \vert \Lambda} \cdot \pi_\Lambda(p)$,
\item $p_{\Theta_1 \vee \Theta_2 \vert \Lambda} = p_{\Theta_1 \vert \Lambda}Ê\cdot p_{\Theta_2 \vert \Lambda}$,
\item $p_{\Theta_1 \vee \Theta_2 \vert \Lambda} = p_1 \cdot p_2$ where $d(p_1) = \Theta_1 \vee \Lambda$ and $d(p_2) = \Theta_2 \vee \Lambda$.,
\item $p  \cdot \pi_\Lambda(p) = \pi_{\Theta_1 \vee \Lambda}(p) \cdot \pi_{\Theta_2 \vee \Lambda}(p)$,
\item $p = p_{\Theta_1 \vert \Lambda} \cdot \pi_{\Theta_2 \vee \Lambda}(p)$,
\item $p_{\Theta_1 \vert \Theta_2 \vee \Lambda} = p_{\Theta_1 \vert \Lambda} \cdot f_{\pi_{\Theta_2 \vee \Lambda}(p)}$,
\item $p_{\Theta_1 \vert \Theta_2 \vee \Lambda} = q \cdot f_{\pi_{\Theta_2 \vee \Lambda}(p)}$, where $d(q) = \Theta_1 \vee \Lambda$.
\end{enumerate}
\end{theorem}

\begin{proof}
We prove $(i) \Rightarrow (i+1)$ for $i=1$ to $7$ and then $(8) \Rightarrow (1)$,

$(1) \Rightarrow (2)$: Using the continuation property of conditionals (\ref{eq:Continuation}) we have, by Theorem \ref{th:CondIndOfPot},
\begin{eqnarray} \label{eq:CojndDecomp}
\lefteqn{p = q_1 \cdot q_2 }Ê \nonumber \\
&&= q_{1\ \Theta_1 \vert \Lambda} \cdot q_{2\ \Theta_2 \vert \Lambda} \cdot \pi_\Lambda(q_1) \cdot \pi_\Lambda(q_2) = q_{1\ \Theta_1 \vert \Lambda} \cdot q_{2\ \Theta_2 \vert \Lambda}  \cdot \pi_\Lambda(p)
\end{eqnarray}
Further, from Theorem \ref{th:CondIndOfPot} we obtain
\begin{eqnarray*}
\pi_{\Theta_1 \vee \Lambda}(p) = q_1 \cdot \pi_\Lambda(q_2) = q_{1\ \Theta_2 \vert \Lambda} \cdot \pi_\Lambda(q_1) \cdot \pi_\Lambda(q_2) = q_{1\ \Theta_1 \vert \Lambda} \cdot \pi_\Lambda(p).
\end{eqnarray*}
Using continuation, we get the equation
\begin{eqnarray*}
\pi_{\Theta_1 \vee \Lambda}(p) = p_{\Theta_1 \vert \Lambda} \cdot \pi_\Lambda(p) = q_{1\ \Theta_1 \vert \Lambda} \cdot \pi_\Lambda(p)
\end{eqnarray*}
and from this we derive $p_{\Theta_1 \vert \Lambda} = q_{1\ \Theta_1 \vert \Lambda} \cdot f_{\pi_\Lambda(p)}$, since by Lemma \ref{le:BasPropOfCiond} $[p_{\Theta_1 \vert \Lambda}] \geq [\pi_\Lambda(p)]$. In the same way we find that $p_{\Theta_2 \vert \Lambda} = q_{2\ \Theta_2 \vert \Lambda} \cdot f_{\pi_\Lambda(p)}$. Thus, from (\ref{eq:CojndDecomp}) we have
\begin{eqnarray*}
p = (q_{1 \Theta_1 \vert \Lambda} \cdot f_{\pi_\Lambda(p)}) \cdot (q_{2\ \Theta_2 \vert \Lambda} \cdot f_{\pi_\Lambda(p)}) \cdot \pi_\Lambda(p) = p_{\Theta_1 \vert \Lambda} \cdot p_{\Theta_2 \vert \Lambda} \cdot \pi_\Lambda(p),
\end{eqnarray*}
since $f_{\pi_\Lambda(p)} \cdot \pi_\Lambda(p) = \pi_\Lambda(p)$.

$(2) \Rightarrow (3)$: By continuation
\begin{eqnarray*}
p = p_{\Theta_1 \vee \Theta_2 \vert \Lambda} \cdot \pi_\Lambda(p)
\end{eqnarray*}
and by (2)
\begin{eqnarray*}
p = p_{\Theta_1 \vert \Lambda}Ê\cdot p_{\Theta_2 \vert \Lambda} \cdot \pi_\Lambda(p).
\end{eqnarray*}
Thus we have
\begin{eqnarray*}
p_{\Theta_1 \vee \Theta_2 \vert \Lambda} \cdot \pi_\Lambda(p) = p_{\Theta_1 \vert \Lambda}Ê\cdot p_{\Theta_2 \vert \Lambda} \cdot \pi_\Lambda(p)
\end{eqnarray*}
Multiplying both sides by $(\pi_\Lambda(p))^{-1}$ and using Lemma \ref{le:BasPropOfCiond} gives (3).

$(3) \Rightarrow (4)$: (4) follows from (3) by taking $p_1 = p_{\Theta_1 \vert \Lambda}$ and $p_2 = p_{\Theta_2 \vert \Lambda}$.

$(4) \Rightarrow (5)$: From (4) using continuation we have
\begin{eqnarray} \label{eq:DecomOfP}
p \cdot \pi_\Lambda(p) = p_{\Theta_1 \vee \Theta_2 \vert \Lambda} \cdot \pi_\Lambda(p) \cdot \pi_\Lambda(p)
= (p_1 \cdot \pi_\Lambda(p)) \cdot (p_2 \cdot \pi_\Lambda(p)).
\end{eqnarray}
Further, again by continuation, and Theorem \ref{th:CondIndOfPot},
\begin{eqnarray*}
\lefteqn{\pi_{\Theta_1 \vee \Lambda}(p) = \pi_{\Theta_1 \vee \Lambda}(p_{\Theta_1 \vee \Theta_2 \vert \Lambda} \cdot \pi_\Lambda(p)) }Ê\\
&&= \pi_{\Theta_1 \vee \Lambda}(p_1 \cdot p_2 \cdot \pi_\Lambda(p)) = p_1 \cdot \pi_\Lambda(p_2 \cdot \pi_\Lambda(p)).
\end{eqnarray*}
It then follows further, see Theorem \ref{th:WeakCombAxiom2},
\begin{eqnarray*}
\pi_{\Theta_1 \vee \Lambda}(p) = p_1 \cdot \pi_\Lambda(p_2) \cdot \pi_\Lambda(p).
\end{eqnarray*}
In the same way we obtain
\begin{eqnarray*}
\pi_{\Theta_2 \vee \Lambda}(p) = p_2 \cdot \pi_\Lambda(p_1) \cdot \pi_\Lambda(p).
\end{eqnarray*}
Furthermore, from Lemma \ref{le:BasPropOfCiond}, item 2, (4) and Theorem \ref{th:CondIndOfPot},
\begin{eqnarray*}
f_{\pi_\Lambda(p)} = \pi_\Lambda(p_{\Theta_1 \vee \Theta_2 \vert \Lambda}) = \pi_\Lambda(p_1 \cdot p_2) = \pi_\Lambda(p_1) \cdot \pi_\Lambda(p_2).
\end{eqnarray*}
This allows us finally to write
\begin{eqnarray*}
\lefteqn{\pi_{\Theta_1 \vee \Lambda}(p) \cdot \pi_{\Theta_2 \vee \Lambda}(p) = (p_1 \cdot \pi_\Lambda(p)) \cdot (p_2 \cdot \pi_\Lambda(p)) \cdot (\pi_\Lambda(p_1) \cdot \pi_\Lambda(p_2))    }Ê\\
&&= (p_1 \cdot \pi_\Lambda(p)) \cdot (p_2 \cdot \pi_\Lambda(p)) \cdot f_{\pi_\Lambda(p)} = (p_1 \cdot \pi_\Lambda(p)) \cdot (p_2 \cdot \pi_\Lambda(p)) Ê\\
&&= p \cdot \pi_\Lambda(p)
\end{eqnarray*}
the last equality is due to (\ref{eq:DecomOfP}).

$(5) \Rightarrow (6)$: Starting with (5) and using continuation, we have
\begin{eqnarray*}
p \cdot \pi_\Lambda(p) = \pi_{\Theta_1 \vee \Lambda}(p) \cdot \pi_{\Theta_2 \vee \Lambda}(p)
= p_{\Theta_1 \vert \Lambda} \cdot \pi_\Lambda(p) \cdot \pi_{\Theta_2 \vee \Lambda}(p).
\end{eqnarray*}
Eliminating $\pi_\Lambda$ on both sides (using Lemma \ref{le:ProjOrder}) yields
\begin{eqnarray*}
p = p_{\Theta_1 \vert \Lambda} \cdot \pi_{\Theta_2 \vee \Lambda}(p)
\end{eqnarray*}

$(6) \Rightarrow (7)$: By continuation and (6) we have
\begin{eqnarray*}
p = p_{\Theta_1 \vert \Theta_2 \vee \Lambda} \cdot \pi_{\Theta_2 \vee \Lambda}(p) = p_{\Theta_1 \vert  \Lambda} \cdot  \pi_{\Theta_2 \vee \Lambda}(p).
\end{eqnarray*}
(7) follows from the rightmost equality by elimination of $ \pi_{\Theta_2 \vee \Lambda}(p)$ (using Lemma \ref{le:ProjOrder}),

$(7) \Rightarrow (8)$: Take $q = p_{\Theta_1 \vert \Lambda}$.

$(8) \Rightarrow (1)$: By continuation and (8)
\begin{eqnarray*}
p = p_{\Theta_1 \vert  \Theta_2 \vee \Lambda} \cdot \pi_{\Theta_2 \vee \Lambda}(p) = q \cdot f_{\pi_{\Theta_2 \vee \Lambda}(p)} \cdot \pi_{\Theta_2 \vee \Lambda}(p) = q \cdot \pi_{\Theta_2 \vee \Lambda}(p),
\end{eqnarray*}
where $d(q) = \Theta_1 \vee \Lambda$ Take now $q_1 = q$ and $q_2 =  \pi_{\Theta_2 \vee \Lambda}(p)$ and then (1) follows.
\end{proof}

We remark that al these results hold also for densities (Section \ref{subsec:AbsContPAS}), as well as in many other valuation algebras, as has been shown in \cite{kohlas03}.

\section{Conditional Independence Structures} \label{sec:LocComp}


\subsection{Markov Trees} \label{subsec:MTrees}

In this section we review a more complex conditional independence structure which plays also an important role in algorithmic issues (see Section \ref{subsec:LocComp}). Most of this material has been developed in \cite{kohlas17}, so that we may refer to this text for proofs. 

Consider a tree $T = (V,E)$ with nodes set $V$ and edges $E \subseteq V^2$, where $V^2$ is the family of two-element subsets of $V$. Let $\Lambda : V \rightarrow \mathcal{F}$ be a labeling of the nodes of the tree with frames. The pair $(T,\Lambda)$ is called a labeled tree. By $ne(v)$ we denote the sets of neighbour nodes of $v$ in the tree, that is $ne(v) = \{w:w \in \{v,w\} \in E\}$. When a node $v$ is eliminated for $T$ together with all edges incident to it, then a family of subtrees $\{T_{v,w} = (V_{v,w},E_{v,w}):w \in ne(v)\}$ remain, where $T_{v,w}$ is the subtree of $T$ containing node $w \in ne(v)$. For any subset $U$ of nodes let
\begin{eqnarray*}
\Lambda(U) = \vee_{v \in U} \Lambda(v).
\end{eqnarray*}
These considerations lead to the definition of \textit{Markov Trees}.

\begin{definition} {Markov Tree:}
A labeled tree $(T,\Lambda)$ with $T =(V,E)$ is called a Markov tree, if for all $v \in V$,
\begin{eqnarray}
\bot \{\Lambda(V_{v,w}):w \in ne(v)\} \vert \Lambda(v).
\end{eqnarray}
\end{definition}

Markov tree have early been identified as important structures for efficient computation with belief functions using Dempster's rule \cite{SSM87,kohlasmonney95}. For computations with probability potentials, structures like join- or junction trees were proposed. In the multivariate setting Markov and join trees are equivalent, but this is no more true in the present more general setting of f.c.f. We refer to \cite{kohlas17} for more on this subject.

The following are two important results on Markov trees:

\begin{theorem} \label{th:SubTreeMarkov}
Let $(T,\Lambda)$ be a Markov tree. Then any subtree is also a Markov tree.
\end{theorem}

\begin{theorem} \label{th:NeigborCondIndep}
Let $(T,\Lambda)$ be a Markov tree. Then for any node $v$ and all nodes $w \in ne(v)$, we have
\begin{eqnarray*}
\Lambda(v) \bot \Lambda(V_{v,w}) \vert \Lambda(w).
\end{eqnarray*}
\end{theorem}
 For the proof of these theorems we refer to \cite{kohlas17}.

As mentioned, Markov trees are important for computational purposes, see Section \ref{subsec:LocComp}

In continuation of the subject of Section \ref{subsec:CondIndPot} we consider now factorizations over conditionally independent frames. We extend Definition \ref{def:CondIndepPot} of Section \ref{subsec:CondIndPot} as follows:

\begin{definition}
Let $\bot \{\Theta_1,\ldots,\Theta_n\} \vert \Lambda$. If $q_1,\ldots,q_n$ are probability potentials such that $d(q_i) = \Theta_i \vee \Lambda$ for $i = 1;\ldots, n$, then we call the probability potentials $q_1$ to $q_n$ conditionally independent given $\Lambda$ and write $\bot \{q_1,\ldots,q_n\} \vert \Lambda$.
\end{definition}
As in the binary case ($n = 2$) this implies the following results on projection:

\begin{theorem}
$\bot \{\Theta_1,\ldots,\Theta_n\} \vert \Lambda$ implies
\begin{enumerate}
\item $\pi_{\Theta_i \vee \Lambda}(q_1 \cdot \ldots \cdot q_n) = q_i \cdot \left( \prod_{j=1,j\not=i}^n \pi_\Lambda(q_j) \right)$,
\item $\pi_\Lambda(q_1 \cdot \ldots \cdot q_n) = \pi_\Lambda(q_1) \cdots \pi_\Lambda(q_n)$.
\end{enumerate}
\end{theorem}

\begin{proof}
Note that $\bot \{\Theta_1,\ldots,\Theta_n\} \vert \Lambda$ implies $\Theta_1 \bot \vee_{j=2}^n \Theta_j \vert \Lambda$, see Theorem \ref{th:ExtQSeparoid}. So, by Theorem \ref{th:CondIndOfPot} 
\begin{eqnarray*}
\pi_\Lambda(q_1 \cdot \ldots \cdot q_n) = \pi_\Lambda(q_1) \cdot \pi_\Lambda(\prod_{j=2}^n q_j).
\end{eqnarray*}
By induction we obtain 
\begin{eqnarray*}
\pi_\Lambda(\prod_{j=2}^n q_j) = \prod_{j=2}^n \pi_\Lambda(q_j),
\end{eqnarray*}
hence item 2 follows. And, again by Theorem \ref{th:CondIndOfPot},
\begin{eqnarray*}
\pi_{\Theta_1 \vee \Lambda}(q_1 \cdot \ldots \cdot q_n) = q_1 \cdot \pi_\Lambda(\prod_{j=2}^n q_j),
\end{eqnarray*}
from which item 1 follows.
\end{proof}

Theorem \ref{th:EquivOfCondIndep} about equivalent formulations of conditional independence extends to the case $n > 2$ in the following way.

\begin{theorem}
Assume $\bot \{\Theta_1,\ldots,\Theta_n\} \vert \Lambda$. Then the following statements are all equivalent:
\begin{enumerate}
\item $p = q_1 \cdot \ldots \cdot q_n$, where $d(q_i) = \Theta_i$ for $i = 1 \ldots,n$,
\item $p = p_{\Theta_1 \vert \Lambda} \cdots p_{\Theta_n \vert \Lambda} \cdot \pi_\Lambda(p)$,
\item $p_{\Theta_1 \vee \cdots \vee \Theta_n \vert \Lambda} = p_{\Theta_1 \vert \Lambda} \cdots p_{\Theta_n \vert \Lambda}$,
\item $p_{\Theta_1 \vee \cdots \vee \Theta_n \vert \Lambda} = p_1 \cdots p_n$ with $d(pq_i) = \Theta_i \vee \Lambda$, $i=1,\ldots,n$,
\item $p \cdot \pi_\Lambda^{n-1}(p) = \pi_{\Theta_1 \vee \Lambda}(p) \cdots \pi_{\Theta_1 \vee \Lambda}(p)$.
\item $p = p_{\Theta_1 \vert \Lambda} \cdot \pi_{\Theta_2 \vee \cdots \vee \Theta_n \vee \Lambda}(p)$.
\item $p_{\Theta_1 \vert \Theta_2 \vee \cdots  \vee \Theta_n \vee \Lambda} = p_{\Theta_1 \vert \Lambda} \cdot f_{\pi_{\Theta_2 \vee \cdots \vee \Theta_n \vee \Lambda}(p)}$.
\item $p_{\Theta_1 \vert \Theta_2 \vee \cdots  \vee \Theta_n \vee \Lambda} = q \cdot f_{\pi_{\Theta_2 \vee \cdots \vee \Theta_n \vee \Lambda}(p)}$, with $d(q) = \Theta_1 \vee \Lambda$.
\end{enumerate}
\end{theorem}

\begin{proof}
The proof of Theorem \ref{th:EquivOfCondIndep} carries easily over to this more general case, or, alternatively, the results may be derived directly from Theorem \ref{th:EquivOfCondIndep}.
\end{proof}

Next, we consider factorizations over Markov trees. So, let $(T,\Lambda)$ be a Markov tree, $T = (V,E)$ and consider a probability potential $p$ sucht that
\begin{eqnarray} \label{eq:MTreeFact}
p = \prod_{v \in V} q_v, \textrm{ with}\ d(q_v)  = \Lambda(v).
\end{eqnarray} \label{eq:MTfact}
As we know, this is a probability potential. Define
\begin{eqnarray} \label{eq:SubTreefact}
p_{v,w} = \prod_{u \in V_{v,w}} q_u
\end{eqnarray}
so that 
\begin{eqnarray*}
p = \prod_{w \in ne(v)} p_{v,w} \cdot q_v,
\end{eqnarray*}
where $d(p_{v,w}) = \Lambda(V_{v,w})$. By the conditional independence condition defining Markov trees, we obtain, by a generalization of Theorem \ref{th:WeakCombAxiom1} (see Section \ref{subsec:AlgOfProbPot})
\begin{eqnarray*}
\pi_{\Lambda(v)}(p) = \prod_{w \in ne(v)} \pi_{\Lambda(v)}(p_{v,w}) \cdot q_v,
\end{eqnarray*}
Using Theorem \ref{th:NeigborCondIndep}, we obtain further
\begin{eqnarray} \label{eq:RecLocComp}
\pi_{\Lambda(v)}(p) = \prod_{w \in ne(v)} \pi_{\Lambda(v)}(\pi_{\Lambda(w)}(p_{v,w})) \cdot q_v,
\end{eqnarray}
We remark that this is a recursive formula to compute the projection $\pi_{\Lambda(v)}(p)$ of the factorization, since the trees $T_{v,w} = (V_{v,w},E_{v,w})$ are still Markov trees and $\pi_{\Lambda(w)}(p_{v,w})$ can be computed in these subtrees in a similar way. This means finally, in order to compute $\pi_{\Lambda(v)}(p)$, we need only to combine probability potentials on nodes of the tree and transport potentials to a neighbouring node. So, this is a \textit{local computation scheme}, generalizing the well-known procedure from multivariate models to the much more general case considered here. 

Let's examine the combination in (\ref{eq:RecLocComp}) a bit more closely. Of course, we may compute this combination sequentially as
\begin{eqnarray*} 
\pi_{\Lambda(v)}(p) = (\ldots((q_v \cdot \pi_{\Lambda(v)}(\pi_{\Lambda(w)}(p_{v,w})) \cdot \pi_{\Lambda(v)}(\pi_{\Lambda(w')}(p_{v,w'}) \cdot \ldots),
\end{eqnarray*}
over any sequence of neighbourg nodes $w,w',\ldots$ of $v$. Each time a transport of potential to node $v$ must be combined with a probability potential on node $v$. Recall the basic formula for this: Let $p_1$ be a potential on frame $\Theta$, $p_2$ on $\Lambda$, then
\begin{eqnarray} \label{eq:CombPlausSingl}
p_1 \cdot \pi_\Theta(p_2)(\theta) = p_1(\theta) \cdot \sum_{\lambda:\tau(\theta) \cap \mu(\lambda) \not= \emptyset} p_2(\lambda),
\end{eqnarray}
where $\tau$ and $\mu$ are the refinings of $\Theta$ and $\Lambda$ to $\Theta \vee \Lambda$. Here $p_1(\theta)$ is simply multiplied with the sum of the $p_2(\lambda)$ over $\lambda$ compatible with $\theta$. Note that this sum is a probability potential, although not normalized. This is important, since it means the whole computational scheme to compute $\pi_{\Lambda(v)}(p)$ runs in the subsystem of set potentials consisting of probability potentials.



\subsection{Local Computation} \label{subsec:LocComp}

Many local computation architectures proposed for multivariate models carry over to the present more general model of probability potentials on general f.c.f. The original paper on local computation in multivariate models with probability potentials is \cite{lauritzenspiegelhalter88}, later work on this subject is to be found in \cite{shafer96,cowell99}. This computational scheme has been generalized to abstract valuation algebras by \cite{shenoyshafer90} and the corresponding computational architecture, based on message passing, has been called \textit{Shenoy-Shafer architecture}, see also \cite{kohlasshenoy00,kohlas03}. This architecture carries directly over to our present case and will be the first scheme shortly described in this section. In the particular case of probability potentials on multivariate frames, division of potentials can be exploited. This can simplify the Shenoy-Shafer-architecture. A first version using division has already been proposed in \cite{lauritzenspiegelhalter88}, a variant thereof called HUGIN-architecture followed, see for instance \cite{shafer96}. Both archtitecures can be used also for some special kind of abstract valuation algebras, called regular or separative valuation algebras. This has first been observed by \cite{lauritzenjensen97} and worked out in \cite{kohlas03,kohlaswilson06}, always in the framework of multivariate models. These computational schemes using division can be adapted to probability potentials on \textit{commutative}  f.c.f. This will be discussed in the next section.

Let $(T,\Lambda)$ be a Markov tree, with $T = (V,E)$ a tree with vertices $V$ and edges $E \subseteq V^2$. Let further $p$ be a factorization over the Markov tree defined by (\ref{eq:MTreeFact}), so that $\pi_{\Lambda(v)}$ for a selected node $v \in V$ can be computed by the recursion (\ref{eq:RecLocComp}). This scheme can be described by a message passing mechanism. Define $p_{v,w}$ as in the previous section by (\ref{eq:SubTreefact}) relative to any pair of vertices. Then let, motivated by (\ref{eq:RecLocComp}),
\begin{eqnarray} \label{eq:OrigMess}
\mu_{w \rightarrow v} =\pi_{\Lambda(v)}(\pi_{\Lambda(w)}(p_{v,w})).
\end{eqnarray}
Define $\eta_w = \pi_{\Lambda(w)}(p_{v,w})$. Then we have, similar to (\ref{eq:CombPlausSingl}),
\begin{eqnarray*}
\mu_{w \rightarrow v}(\theta) = \sum_{\lambda:\tau(\theta) \cap \nu(\lambda) \not= \emptyset} \eta_w(\lambda),
\end{eqnarray*}
for all $\theta \inÊ\Lambda(v)$, if $\tau$ and $\nu$ are the refinings of $\Lambda(v)$ and $\Lambda(w)$ to $\Lambda(v) \vee \Lambda(w)$ resepectively. 

The recursive computational scheme of the previous section can now be described in terms of messages as follows: 
\begin{enumerate}
\item There is always at least one leaf node $w$ in the tree, which is incident to only one edge, hence with a single neighbour $v$. Then $\eta_w = q_w$ and the message $\mu_{w \rightarrow v}$ can be computed.
\item Once a node $w$ has received message form all its neighbors $u$, except a node $v$, then it can compute
\begin{eqnarray*}
\eta_w = q_w \cdot \prod_{u \in ne(w),u \not= v} \mu_{u \rightarrow w}
\end{eqnarray*}
and it can send the message $\mu_{w \rightarrow v}$ to its neighbor $v$.
\item The last node $v$, called \textit{root node}, which received messages from all its neighbors computes
\begin{eqnarray*}
\eta_v = \pi_{\Lambda(v)}(p) = q_v \cdot \prod_{w \in ne(v)} \mu_{w \rightarrow v}.
\end{eqnarray*}
\end{enumerate}
Note that in this procedure, we have $\eta_w = \pi_{\Lambda_w}(p_{v,w})$, when $w$ sends a message to node $v$. This procedure is called the \textit{collect algorithm}.

In the Shenoy-Shafer architecture it is proposed to store a message $\mu_{w \rightarrow v}$ on the edge $\{v,w\}$ so that it can be reused for a second phase after running the collect algorithm. In fact the root node $v$ can send messages
\begin{eqnarray} \label{eq:distmess}
\mu_{v \rightarrow w} = \pi_{\Lambda(w)}(q_v \cdot \prod_{u \in ne(v),u \not= w} \mu_{u \rightarrow v})
\end{eqnarray}
to all its neighbors $w$ (note that the messages $\mu_{u \rightarrow v}$ are stored on the edges $\{u,w\}$ in the collect phase). Then all these neighbors can compute
\begin{eqnarray*}
\eta_w = \pi_{\Lambda(w)}(p) = q_w \cdot \prod_{u \in ne(w)} \mu_{u \rightarrow w},
\end{eqnarray*}
and send further messages to their other neighbours, different from $v$, etc. until all nodes have computed $\pi_{\Lambda(w)}(p)$. This is called the \textit{distribute algorithm}. Collect and distribute algorithm constitute what is called the Shenoy-Shafer architecture to compute \textit{all} projections $\pi_{\Lambda(w)(p)}$ for a join tree factorization. 

This is called a local computation procedure, because the essential operations of combination of probability potentials are always carried out on a local domain $\Lambda(w)$. This is much more efficient than the naive approach in which first the combination $q_v \cdot q_w \cdot \ldots$ is computed on ever growing domains $\Lambda(v) \vee \Lambda(w) \vee \ldots$. Still there are some inefficiencies in this procedure, since if a node has more than three edges incident some sub-combinations of messages have to be computed several times. Therefore in \cite{shenoy97} more special, binary Markov trees in the multivariate setting have been proposed, which avoid this redundant combinations. It seems possible to extend this approach to the present more general Markov trees. Another method to avoid the redundant combinations is proposed in the next section. Computation with the Shenoy-Shafer architecture can also be applied to compute projections of combinations of set potentials \cite{kohlas03}.

The procedures described apply only to factorizations of potentials over a Markov tree. There remain a number of practical questions: If a combination of potentials is given, how can we decide whether it is a factorization over a Markov tree. If not, what then? In the multivariate case, the concept of covering trees is used. This applies also in our more general case. However, how do we find such a Markov tree? In the multivariate case this is done by successive variable elimination (see for instance \cite{kohlas03}, where further references are given). This is not possible for f.c.f in general. These, and other questions remain open so far.


\subsection{Local Computation in Commutative Frames} \label{subsec:Division}

According to Section \ref{subsec:CommFcF} the probability potentials on a commutative f.c.f form a valuation algebra satisfying axioms B0 to B5. But they are still embedded in the algebra of bpa. Further in a commutative f.c.f $(\mathcal{F},\leq)$ forms a lattice and for all frames $\Theta$ and $\Lambda$  we have $\Theta \bot \Lambda \vert \Theta \wedge \Lambda$. This permits to compute the transport operation $t_\Lambda$ for a potential with domain $\Theta$ within the algebra of set potentials as
\begin{eqnarray} \label{eq:TranspComFcF}
\pi_\Lambda(p) = \pi_{\Theta \wedge \Lambda}(p) \cdot \mathbf{1}_\Lambda.
\end{eqnarray}
We refer to \cite{kohlas03} for a derivation of this result. This in turn allows to write (\ref{eq:OrigMess}) as follows
\begin{eqnarray*}
\mu_{w \rightarrow v} =\pi_{\Lambda(w) \wedge \Lambda(v)}(\pi_{\Lambda(w)}(p_{v,w})) \cdot \mathbf{1}_{\Lambda(v)},
\end{eqnarray*} 
and in fact, the factor $\mathbf{1}_{\Lambda(v)}$ may be dropped in the Shenoy-Shafer architecture, since this message is alway combined wth $q_v$ and is thus absorbed by this term. So, in the case of commutative f.c.f let's define the messages as 
\begin{eqnarray} \label{eq:CommFcFmess}
\mu_{w \rightarrow v} =\pi_{\Lambda(w) \wedge \Lambda(v)}(\pi_{\Lambda(w)}(p_{v,w})),
\end{eqnarray} 
Then all the rest of collect and distribute algorithm of the Shenoy-Shafer architecture remains as in the general case. Note that as before all the messages remain probability potentials, that is the whole computational scheme runs within the valuation algebra of probability potentials on commutative f.c.f. Therefore, in this case computations may be somewhat simplified using division in the framework of the valuation algebra of probability potentials.

This is based on the fact that probability potentials have inverses as discussed in Section \ref{subsec:CondAndCont}, so that potentials like messages $\mu_{u \rightarrow w}$ can be divided out and this can be used to avoid redundant combinations.  Equipped with this operation of division, we may construct variants of both Lauritzen-Spiegelhalter and HUGIN architectures for potentials on a commutative f.c.f, see also \cite{kohlas03} for the multivariate case. In the first case, during the collect phase, instead of storing the message $\mu_{w \rightarrow v}$ on the edges $\{w,v\}$, rather its inverse $\mu_{w \rightarrow v}^{-1}$ is stored there. In both architectures the collect algorithm is essentially as in the Shenoy-Shafer architecture. That is, a node $u$ has associated the valuation
\begin{eqnarray} \label{eq:HUGINStore}
\eta_u = q_u \cdot \prod_{v \in ne(u),v \not= w} \mu_{v \rightarrow u}
\end{eqnarray}
just before it sends the message $\mu_{u \rightarrow w}$ according to (\ref{eq:CommFcFmess}) to node $w$. However, in the Lauritzen-Spiegelhalter architecture in the sending node $u$ this message is divided out, so that now the valuation
\begin{eqnarray} \label{eq:LSStore}
\eta_u := \eta_u \cdot \mu_{u \rightarrow w}^{-1} = q_u \cdot \prod_{v \in ne(u),v \not= w} \mu_{v \rightarrow u} \cdot \mu_{u \rightarrow w}^{-1}
\end{eqnarray}
is stored in this node. In contrast, in the HUGIN architecture during collect, the inverse of the message $\mu_{u \rightarrow w}$ is stored on the edge $\{u,w\}$ linking nodes $u$ and $w$ rather than divided out in node $u$. 

In the distribute phase, starting with the root node $v$, any node $w$ contains
\begin{eqnarray*}
\eta_w = q_w \cdot \prod_{n \in ne(w)} \mu_{n \rightarrow w}
\end{eqnarray*}
just before it sends the distribute message $\mu_{w \rightarrow u} = \pi_{\Lambda(w) \wedge \Lambda(u)}(\eta_w)$ to its neighbour $u$. We show by induction that this gives $\eta_w = \pi_{\Lambda(w)}(p)$. This holds for the root according to the discussion of the collect algorithm in the previous section. Assume it holds for $w$. Then in the Lauritzen-Spiegelhalter architecture, by the assumption of induction, the message sent to node $u$ is 
\begin{eqnarray*}
\lefteqn{\pi_{\Lambda(w) \wedge \Lambda(u)}(p) = \pi_{\Lambda(w) \wedge \Lambda(u)}(q_w \cdot \prod_{n \in ne(w)} \mu_{n \rightarrow w}) }\\
&&= \pi_{\Lambda(w) \wedge \Lambda(u)}(q_w \cdot \prod_{n \in ne(w),n \not= u} \mu_{n \rightarrow w}) \cdot \mu_{u \rightarrow w} \\
&&\mu_{w \rightarrow u} \cdot \mu_{u \rightarrow w}
\end{eqnarray*}
by the Combination Axiom B5, since we have $d(\mu_{u \rightarrow w}) = \Lambda(u) \wedge \Lambda(w)$. So, we have in node $u$ after combining the incoming message with the node store $\eta_u$ (\ref{eq:LSStore}) 
\begin{eqnarray*}
\lefteqn{ q_u \cdot \prod_{n \in ne(u),n \not= w} \mu_{v \rightarrow u} \cdot \mu_{u \rightarrow w}^{-1} \cdot \mu_{w \rightarrow u} \cdot \mu_{u \rightarrow w} }\\
&&=\pi_{\Lambda(u)}(p) \cdot f_{\mu_{u \rightarrow w}} = \pi_{\Lambda(u)}(p).
\end{eqnarray*}
The last equality holds because the support of $\mu_{u \rightarrow w}$ is larger than the one of $\pi_{\Lambda(u)}(p)$, see Lemma \ref{le:ProjOrder},
\begin{eqnarray*}
supp(\prod_{n \in ne(u)} \mu_{v \rightarrow u}) \subseteq supp(\mu_{u \rightarrow w}).
\end{eqnarray*}
This confirms the claim in the case of the Lauritzen-Spiegelhalter architecture. In the case of the HUGIN architecture, the situation is similar. The message $\mu_{w \rightarrow u}$ passes through the edge $\{w,u\}$, where it is combined with the valuation $\mu_{u \rightarrow w}^{-1}$ stored there, before the combination $\mu_{w \rightarrow u} \cdot \mu_{u \rightarrow w}^{-1}$ is combined with the valuation $\eta_u$ (\ref{eq:HUGINStore}) stored in node $u$. Then by the same argument we see that this gives again the valuation $\pi_{\Lambda(u)}(p)$. Therefore, both architecture with division give the same correct results. Both architectures apply also for computations in the valuation algebra of densities, see \cite{kohlas03}. These computational schemes can however not be applied in the case of a general f.c.f since no meets between domains need to exist.


\section{Most Probable Configuration} \label{sec:MostProb}

\subsection{Max/Product-Algebra of Potentials}

In this section we address the problem of finding the most probable configuration of a potential $p$ on a domain $\Theta$, that is, to determine an element or the elements of $\Theta$. which maximize $p(\theta)$. If $p$ is given explicitly, say in for a list of pairs $(\theta,p(\theta))$ it is no big problem, even if $\Theta$ has a big cardinality. The problem changes and becomes more important, when $p$ is given implicitly as a combination
\begin{eqnarray}
p = p_1 \cdot \ldots \cdot p_n
\end{eqnarray}
of a large number of potentials each with a relatively small domain. Then the case is complicated because $p$ is not given explicitly, but has to be computed. This is a realistic scenario and the solution of the maximization problem in this framework will be discussed in this section. In the case of a multivariate model this problem has been solved by local computation, corresponding to dynamic programming, see for instance \cite{shenoy91b,shenoy96}. Here we show that their approach extends to our much more general case of potentials on a family of compatible frames.

The starting point is the observation that there is an information algebra of probability potentials on a f.c.f $(\mathcal{F},\mathcal{R})$ associated with this maximization problem. Labeling and Combination are defined as before, transport however is now maximization. More precisely: Let $\Phi_\Theta$ denote the set of all potentials on frame $\Theta \in \mathcal{F}$, and 
\begin{eqnarray}
\Phi = \bigcup_{\Theta \in \mathcal{F}} \Phi_\Theta.
\end{eqnarray}
Then we have the following operations.

\begin{enumerate}
\item \textit{Labeling:} $d : \Phi \rightarrow \mathcal{F}$, defined by $p \mapsto d(p) = \Theta$ if $p \in \Phi_\Theta$.
\item \textit{Combination:} $\cdot : \Phi \times \Phi \rightarrow \Phi$, defined by $(p_1,p_2) \mapsto p_1 \cdot p_2$, where for $\theta \in d(p_1) \vee d(p_2)$,
\begin{eqnarray} \label{eq:CombOfPot}
 p_1 \cdot p_2(\theta) = p_1(t_{\Theta_1}(\theta))p_2(t_{\Theta_2}(\theta)).
 \end{eqnarray} \label{eq:TranspOfSePot}
\item \textit{Transport:} $t : \Phi \times \mathcal{F} \rightarrow \Phi$, defined by $(p,\Lambda) \mapsto t_\Lambda(p)$, where for $\lambda \in \Lambda$, $d(p) = \Theta$,
\begin{eqnarray}
t_\Lambda(p)(\lambda) = \max_{\theta \in \Theta:\theta \sim \lambda} p(\theta),
\end{eqnarray}
\end{enumerate}

So, here transport means to maximize $p(\theta)$ over all elements in $\Theta$, which are compatible with the element $\lambda \in \Lambda$. We have to verify that this algebraic structure indeed satisfies all of the axioms A0 to A6 of an information algebra. Axioms A0 to A3 and A6 are obvious. Here are the two main results which show that axioms A4 and A5 holds too..

\begin{theorem} \label{th:MaxTrAxiom}
Assume $\Theta \bot \Lambda \vert \Lambda_1$ and $d(p) = \Theta$. then we have
\begin{eqnarray} \label{eq:MaxTrAxiom}
\max_{\theta \in \Theta:\theta \sim \lambda} p(\theta)
= \max_{\lambda_1 \in \Lambda_1: \lambda_1 \sim \lambda}(\max_{\theta \in \Theta:\theta \sim \lambda_1} p(\theta))
\end{eqnarray}
\end{theorem}

\begin{proof}
By Lemma \ref{CondIndepRes1} we have under the conditions of the theorem that $\theta \sim \lambda_1$ implies $\theta \sim \lambda$. This implies that
\begin{eqnarray*}
\max_{\theta \sim \lambda} p(\theta) \geq \max_{\theta \sim \lambda_1} p(\theta)),
\end{eqnarray*}
for all $\lambda_1 \in R_\lambda(\Lambda_1)$, hence 
\begin{eqnarray*}
\max_{\theta \sim \lambda} p(\theta)
\geq \max_{\lambda_1 \sim \lambda}(\max_{\theta \sim \lambda_1} p(\theta)).
\end{eqnarray*}

For any $\lambda \in \Lambda$, let $\hat{\theta}(\lambda)$ be a maximizing value of the left hand side of (\ref{eq:MaxTrAxiom}), that is 
\begin{eqnarray*}
p(\hat{\theta}(\lambda)) = \max_{\theta \sim \lambda} p(\theta).
\end{eqnarray*}
Then $\lambda$ and $\hat{\theta}(\lambda)$ are compatible. By Lemma \ref{CondIndepRes2} there is then an element in $\Lambda_1$, say $\lambda_1(\lambda)$ so that $\lambda_1(\lambda) \sim \lambda$ and $\hat{\theta}(\lambda) \sim \lambda_1(\lambda)$. Then we have 
\begin{eqnarray*}
\max_{\theta \sim \lambda_1(\lambda)} p(\theta) \geq p(\hat{\theta}(\lambda)),
\end{eqnarray*}
and therefore also
\begin{eqnarray*}
\max_{\lambda_1 \sim \lambda}(\max_{\theta \sim \lambda_1} p(\theta)) \geq
p(\hat{\theta}).
\end{eqnarray*}
Since the inverse equality has been shown above, this proves (\ref{eq:MaxTrAxiom}).
\end{proof}

This theorem shows that the Transport Axiom A4 holds. The next theorem shows that the Combination Axiom A5 holds too:

\begin{theorem} \label{th:MaxComAxiom}
Assume $\Theta_1 \bot \Theta_2 \vert \Lambda$ and $d(p_1) = \Theta_1$. $d(p_2) = \Theta_2$, then we have
\begin{eqnarray} \label{eq:MaxComrAxiom}
\max_{\theta \in \Theta_1 \vee \Theta_2: \theta \sim \lambda} p_1 \cdot p_2(\theta)
= \max_{\theta_1 \in \Theta_1:\theta_1 \sim \lambda} p_1(\theta_1) \cdot \max_{\theta_2 \in \Theta_2:\theta_2 \sim \lambda} p_2(\theta_2).
\end{eqnarray}
\end{theorem}
\begin{proof}
Note that 
\begin{eqnarray}
\max_{\theta \in \Theta_1 \vee \Theta_2: \theta \sim \lambda} p_1 \cdot p_2(\theta) = \max_{\theta \in \Theta_1 \vee \Theta_2):\theta \sim \lambda} p_1(t_{\Theta_1}(\theta))p_2(t_{\Theta_2}(\theta)).
\end{eqnarray}
Assume this maximum is attained by an element $\hat{\theta} \in R_\lambda(\Theta_1 \vee \Theta_2)$, so that
\begin{eqnarray}
\max_{\theta \in \Theta_1 \vee \Theta_2: \theta \sim \lambda} p_1 \cdot p_2(\theta) = p_1(t_{\Theta_1}(\hat{\theta}))p_2(t_{\Theta_2}(\hat{\theta})).
\end{eqnarray}
Now, by Lemma \ref{CondIndepRes3}, we have $t_{\Theta_1}(\hat{\theta}) \sim \lambda$. Suppose there is a $\theta_1 \in R_\lambda(\Theta_1)$ such that $p_1(\theta_1) > p_1(\hat{\theta})$. Then, $(\theta_1,t_{\Theta_2}(\hat{\theta})) \in R_\lambda(\Theta_1) \times R_\lambda(\Theta_2) = R_\lambda(\Theta_1,\Theta_2)$ and, by the same Lemma, there is an element $\theta \in R_\lambda(\Theta_1 \vee \Theta_2)$ such that $t_{\Theta_1}(\theta) = \theta_1$ and $t_{\Theta_2}(\theta) = t_{\Theta_2}(\hat{\theta})$. But then
\begin{eqnarray*}
p_1(\theta_1)p_2(t_{\Theta_2}(\hat{\theta})) = p_1(t_{\Theta_1}(\theta))p_2(t_{\Theta_2}(\theta)) > p_1(t_{\Theta_1}(\hat{\theta}))p_2(t_{\Theta_2}(\hat{\theta})).
\end{eqnarray*}
This is a contradiction. Therefore we conclude that
\begin{eqnarray*}
p_1(t_{\Theta_1}(\theta)) = \max_{\theta_1 \in \Theta_1:\theta_1 \sim \lambda} p_1(\theta_1)
\end{eqnarray*}
and in the same way we see that
\begin{eqnarray*}
p_2(t_{\Theta_2}(\theta)) = \max_{\theta_2 \in \Theta_2:\theta_2 \sim \lambda} p_2(\theta_2).
\end{eqnarray*}
This concludes the proof
\end{proof}

If the f.c.f $(\mathcal{F},\mathcal{R})$ is \textit{commutative}, the max/product algebra of potentials becomes as the original algebra of potentials a \textit{valuation algebra}, see Section \ref{subsec:CommFcF}. This is in particular the case for the mutivariate model. As mentioned above, this is the case of dynamic programming treated in the paper \cite{shenoy96}. We now have seen that this approach generalizes to to the much more general case of f.c.f. In the next section, this general case of dynamic programming will be further developed and completed. In particular it will be embedded into the local computation approach on Markov trees.

It should also be remarked that the same approach serves to find \textit{most plausible} solutions with regard to likelihood functions associated with the algebra of bpa. This too generalizes from the case of multivariate models \cite{shenoy96}.

\subsection{Solution Construction}

Let $p$ be a potential on some frame $\Theta \in \mathcal{F}$ and $\Lambda$ any other frame of the f.c.f. Consider then the family of maximization problems
\begin{eqnarray*}
\max_{\theta \in \Theta:\theta \sim \lambda} p(\theta) \textrm{ for}\ \lambda \in \Lambda.
\end{eqnarray*}
For any $\lambda$ of the frame $\Lambda$ there are one or several elements $\theta$ in $\Theta$ for which the maximum is attained. Let $s_p^\Lambda(\lambda)$ denote the set of maximizing elements for $\max_{\theta \in \Theta:\theta \sim \lambda} p(\theta)$, that is,
\begin{eqnarray*}
p(\hat{\theta}) = \max_{\theta \in \Theta:\theta \sim \lambda} p(\theta) \textrm{ for any}\ \hat{\theta} \in s_p^\Lambda(\lambda) \textrm{ and}\ \lambda \in \Lambda.
\end{eqnarray*}
We call the map $s_p^\Lambda : \Lambda \rightarrow 2^\Theta$ \textit{solution} for potential $p$ relative to $\Lambda$ and $s_p^\Lambda(\lambda)$ the corresponding \textit{solution sets}. This map can be extended in the usual way to a map $s_p^\Lambda : 2^\Lambda \rightarrow 2^\Theta$.

The Axioms of Transport and Combination for the max/product algebra of potentials imply the following theorem.

\begin{theorem} \label{th:SolProp}
\
\begin{enumerate}
\item Assume $\Theta \bot \Lambda \vert \Lambda_1$ and $d(p) = \Theta$. Then for all $\lambda \in \Lambda$,
\begin{eqnarray} \label{eq:CompOfSol}
s_p^\Lambda(\lambda) = s_p^{\Lambda_1}(s_{t_{\Lambda_1}(p)}^\Lambda(\lambda)).
\end{eqnarray}
\item Assume $\Theta_1 \bot \Theta_2 \vert \Lambda$ and $d(p_1) = \Theta_1$, $d(p_2) = \Theta_2$. Then  for all $\lambda \in \Lambda$,
\begin{eqnarray} \label{eq:LocalityOfSol}
s_{p_1 \cdot p_2}^\Lambda(\lambda) = \tau_1(s_{p_1}^\Lambda(\lambda)) \cap \tau_2(s_{p_2}^\Lambda(\lambda)),
\end{eqnarray}
where $\tau_1$ and $\tau_2$ are the refinings of $\Theta_1$ and $\Theta_2$ to $\Theta_1 \vee \Theta_2$.
\end{enumerate}
\end{theorem}

\begin{proof}
1.) Consider $\hat{\theta} \in s_p^\Lambda(\lambda)$. Then, by (\ref{eq:MaxTrAxiom}) we have 
\begin{eqnarray*}
p(\hat{\theta}) = \max_{\theta \sim \lambda} p(\theta) = \max_{\lambda_1 \sim \lambda}(\max_{\theta \sim \lambda_1} p(\theta)).
\end{eqnarray*}
This shows that, if $\hat{\lambda}_1 \in s_{t_{\Lambda_1}(p)}^\Lambda(\lambda)$, then $\hat{\theta} \in s_p^{\Lambda_1}(\hat{\lambda})$. From this it follows that $\hat{\theta} \in s_p^{\Lambda_1}(s_{t_{\Lambda_1}(p)}^\Lambda(\lambda))$.

On the other hand, assume $\hat{\theta} \in s_p^{\Lambda_1}(s_{t_{\Lambda_1}(p)}^\Lambda(\lambda))$. Then by (\ref{eq:MaxTrAxiom}) we have
\begin{eqnarray*}
p(\hat{\theta}) = \max_{\theta \sim \lambda} p(\theta)
\end{eqnarray*}
and therefore $\hat{\theta}                                                                                 \in s_p^\Lambda(\lambda)$. This proves (\ref{eq:CompOfSol}).

2.) Next, let $\theta \in s_{p_1 \cdot p_2}^\Lambda(\lambda)$. Then, by the proof of Theorem \ref{th:MaxComAxiom} we have $t_{\Theta_1}(\theta) \in s_{p_1}^\Lambda(\lambda)$ and $t_{\Theta_2}(\theta) \in s_{p_2}^\Lambda(\lambda)$. Further $\theta \in \tau_1(t_{\Theta_1}(\theta)) \cap  \tau_2(t_{\Theta_2}(\theta))$, hence $\theta \in \tau_1(s_{p_1}^\Lambda(\lambda) \cap \tau_2(s_{p_2}^\Lambda(\lambda)$. If, on the other hand, $\theta \in \tau_1(s_{p_1}^\Lambda(\lambda)) \cap \tau_2(s_{p_2}^\Lambda(\lambda))$, then there are elements $\theta_1 \in s_{p_1}^\Lambda(\lambda)$ and $\theta_2 \in s_{p_2}^\Lambda(\lambda)$ such that $\{\theta\} = \tau_1(\theta_1) \cap  \tau_2(\theta_2)$ and then by (\ref{eq:MaxComrAxiom}) we have $\theta \in s_{p_1 \cdot p_2}^\Lambda(\lambda)$. This proves (\ref{eq:LocalityOfSol}).
\end{proof}

This theorem can be applied to compute solutions of the most probable configuration problem for a combination of potentials,
\begin{eqnarray*}
\max_{\theta \in \Theta} \ p_1 \cdot \ldots \cdot p_n(\theta), \textrm{ where}\ \theta \in d(p_1) \vee \ldots \vee d(p_n).
\end{eqnarray*}
This is in particular the case if $p = p_1 \cdot \ldots \cdot p_n$ is a Markov tree factorization (see Section \ref{subsec:MTrees}). So, let $(T,\Lambda)$ with $T = (V,E)$ be a Markov tree and 
\begin{eqnarray*}
p = \prod_{v \in V} p_v, \textrm{ with}\ d(p_v) = \Lambda(v).
\end{eqnarray*}
Let us denote the generic elements of $\Lambda_v$ by $\lambda_v$. Since the max/product algebra is an information algebra, formula (\ref{eq:RecLocComp}) holds in this algebra too. It reads in this case, if $\theta \in \vee_{v \in V} d(p_v) = \Theta$,
\begin{eqnarray*}
\max_{\theta \sim \lambda_v} p(\theta) = p_v \cdot \prod_{w \in ne(v)} \max_{\lambda_w \sim \lambda_v} (\max_{\theta_{v,w} \sim \lambda_w} p_{v,w}(\theta_{v,w})).
\end{eqnarray*}
where (see (\ref{eq:SubTreefact}))
\begin{eqnarray*}
p_{v,w} = \prod_{v \in V_{v,w}} p_v
\end{eqnarray*}
and where $\Theta_{v,w}$ is the the domain of $p_{v,w}$. Let
\begin{eqnarray*}
q_{v,w}(\lambda_w) = \max_{\theta_{v,w} \sim \lambda_w} p_{v,w}(\theta_{v,w})
\end{eqnarray*}
and
\begin{eqnarray*}
q_v(\lambda_v) = p_v(\lambda_v) \cdot \prod_{w \in ne(v)} \max_{\lambda_w \sim \lambda_v} q_{v,w}(\lambda_w).
\end{eqnarray*}
Note that the potentials $q_{v,w}$ and $q_v$ can be computed by the collect algorithm in the max/product algebra (see Section \ref{subsec:LocComp}). At the end of the collect phase, we have
\begin{eqnarray*}
\max_{\theta \in \Theta} p(\theta) = \max_{\lambda_v \in \Lambda_v} q_v(\lambda_v).
\end{eqnarray*}
In this way the \textit{value} of the maximization problem is obtained by local computation. Note that this final maximization corresponds to projection to the bottom frame $\mathcal{E}$ of the f.c.f, that is $t_{\mathcal{E}}(p) = t_{\mathcal{E}}(q_v)$.

However, we want to compute not only the value, but \textit{solution configurations} of the maximization problem. So, we may determine in the last step the solutions $s_{q_v}^{\mathcal{E}}$, that is the solution set $s_{q_v}^{\mathcal{E}}(e)$ of $t_{\mathcal{E}}(q_v)$. Further, during collect, assume that we determine and store the solution $s_{q_{v,w}}^{\Lambda_v}$ of the transport $t_{\Lambda_v}(q_{v,w})$ when solving the maximization problems $\max_{\lambda_w \sim \lambda_v} q_{v,w}(\lambda_w)$. This solution is represented by solution sets $s_{q_{v,w}}^{\Lambda_v}(\lambda_v)$ for all $\lambda_v \in \Lambda_v$. If we store these solution sets during collect, then we may obtain solution sets recusrively
\begin{eqnarray*}
s_{t_{\mathcal{E}}(t_{\Lambda_v}(q_{v.w}))}^{\mathcal{E}}(e) = s_{q_{v,w}}^{\Lambda_v}(s_{q_v}^{\mathcal{E}}(e)))
\end{eqnarray*}
for $t_{\mathcal{E}}(t_{\Lambda_v}(q_{v.w}))$. 

In this way, we get first the solutions for the domain $\Lambda_v$, and then for all the neighbors $w \in ne(v)$. Since each of the substrees $T_{v,w}$ is still a Markov tree, the procedure can be repeated until the solutions on all domains $\Lambda_u$ for all nodes $u \in V$ are found. From these partial solutions, the overall solution in $\bigvee_{v \in V} \Lambda_v$ may be constructed using the appropriate refinings. This, of course, is only a sketch of a dynamic programming approach to compute the most probable configuration in a f.c.f Detaisl, leading to actual algorithms need to be worked out further.

\section{Conclusion}

Probabilistic Argumentation Systems are an alternative to the popular Bayesian or probability networks. They are more related to functional or logical modeling of uncertain situations than causal modeling, the domain of Bayesian networks. In contrast to probabilistic networks they are not restricted to multivariate models, but can be used with partition models or more generally families of compatible frames. From a computational point of view, inference with PAS is based on similar or identical algebraic structures as Bayesian networks, namely valuation and information algebras, allowing for local computation schemes. In fact, the algebraic structures are somewhat more natural and more easily interpreted for PAS than for probabilistic networks using conditional distributions.

\bibliography{tcslit}
\bibliographystyle{authordate3}


\end{document}